\documentclass[11pt,a4paper]{article}
\usepackage{graphicx}
\RequirePackage[hidelinks]{hyperref}

\usepackage[margin=2.5cm]{geometry}

\usepackage{amsmath}
\usepackage{amssymb}
\usepackage{amsthm}
\usepackage[utf8]{inputenc}
\usepackage[normalem]{ulem} 
\usepackage{xcolor}
\usepackage{dsfont} 
\usepackage{verbatim}  
\usepackage{natbib}

\usepackage{array}
\newcolumntype{C}[1]{>{\centering\arraybackslash}p{#1}}

\newtheorem{de}{Definition}

\newtheorem{lemma}[de]{Lemma}
\newtheorem{theo}[de]{Theorem}
\newtheorem{corollary}[de]{Corollary}
\newtheorem{remark}[de]{Remark}

\DeclareMathOperator*{\argmin}{argmin\,}
\newcommand{\RR}{\mathbb{R}}
\newcommand{\Rd}{\RR^d}
\newcommand{\tss}{\mathrm{TSS}}
\newcommand{\mss}{\mathrm{MSS}}
\newcommand{\rss}{\mathrm{RSS}}
\newcommand{\ssa}{\mathrm{SSa}}
\newcommand{\ssb}{\mathrm{SSb}}
\newcommand{\ssi}{\mathrm{SSi}}
\newcommand{\inpto}{\stackrel{p}{\longrightarrow}}
\newcommand{\inlawto}{\stackrel{\mathcal{D}}{\longrightarrow}}
\newcommand{\mcn}{\mathcal{N}}
\newcommand{\tn}{\tilde{n}}
\newcommand{\tN}{\tilde{N}}
\newcommand{\tX}{\tilde{X}}
\newcommand{\eps}{\varepsilon}
\newcommand{\abs}[1]{\lvert #1 \rvert}
\newcommand{\norm}[1]{\lVert #1 \rVert}
\newcommand{\mcr}{\mathcal{R}}
\newcommand{\mcx}{\mathcal{X}}

\newcommand{\EE}{\mathbb{E}}
\newcommand{\NN}{\mathbb{N}}
\newcommand{\ZZ}{\mathbb{Z}}

\newcommand{\mfnfin}{\mathfrak{N}_{\mathrm{fin}}}

\newcommand{\nin}{\noindent}
\newcommand{\dtt}{d_{\mathrm{TT}}}
\newcommand{\drtt}{d_{\mathrm{RTT}}}
\newcommand{\Leb}{\mathrm{Leb}}
\newcommand{\strauss}{\mathrm{Strauss}}
\newcommand{\csr}{\mathrm{CSR}}

\begin{document}

\title{ANOVA for Data in Metric Spaces,\\ with Applications to Spatial Point Patterns}
\author{Raoul M\"uller\footnote{Work supported by DFG RTG 2088.}\ \footnote{raoul.mueller@uni-goettingen.de}\ \footnote{Institute for Mathematical Stochastics, University of G\"ottingen, 37077 G\"ottingen, Germany.} \and Dominic Schuhmacher\footnotemark[3] \and Jorge Mateu\footnote{Work partially funded by Ministry of Science and Innovation with grant PID2019-107392RB-I00, and by grant UJI-B2021-37 from University Jaume I.}\ \footnote{Department of Mathematics, University Jaume I, 12071 Castell\'on, Spain.}  \\
}
\maketitle

\begin{abstract}
We give a review of recent ANOVA-like procedures for testing group differences based on data in a metric space and present a new such procedure.
Our statistic is based on the classic Levene's test for detecting differences in dispersion. It uses only pairwise distances of data points and and can be computed quickly and precisely in situations where the computation of barycenters (``generalized means'') in the data space is slow, only by approximation or even infeasible. We show the asymptotic normality of our test statistic and present simulation studies for spatial point pattern data, in which we compare the various procedures in a 1-way ANOVA setting. As an application, we perform a 2-way ANOVA on a data set of bubbles in a mineral flotation process.
\end{abstract}

\noindent

\section{Introduction}

Real-world statistical data is often not Euclidean, involving components that are most suitably analyzed in a more complicated space. Examples include spaces of point patterns and more general subsets, trees and more general graphs, functions and images.

In recent years a number of methods have been proposed for analyzing group differences of such data by generalizing classical analysis of variance (ANOVA) ideas to more complex data spaces. Examples include \cite{cuevas2004anova} for functional data, \cite{huckemann2009intrinsic} for data on Riemannian manifolds and \cite{ramon2016new} for point pattern data.
A common feature of the underlying spaces is that there is typically a more or less natural concept of distance between data points available. In addition to the more obvious choices of distances on function spaces and Riemannian manifolds, suitable metrics for tree spaces, graph spaces and point pattern spaces can be found in \cite{billera2001geometry}, \cite{ginestet2017hypothesis} and \cite{muller2020metrics}, respectively.

In the present paper we focus on generalized ANOVA-procedures for metric spaces without using any more special structure of the space. There is a number of preceding articles that work in similar generality.

\cite{anderson2001new} proposes to perform ANOVA based on pairwise dissimilarities of ob\-ser\-va\-tions rather than Euclidean distances between observations and their group means, and introduces the name PERMANOVA for this procedure.
While not directly referring to any more abstract spaces than $\Rd$, that article clearly discusses the abstract template of doing non-Euclidean ANOVA without using a centroid object.
We discuss this further in Subsection~\ref{subsec:Anderson}. \cite{anderson2006distance}
proposes multidimensional scaling followed by a Levene's test (using the cen\-troid object in the principal coordinate space) for detecting differences of within-group dis\-per\-sions (scatter, variability); this is referred to as PERMDISP, see \cite{anderson2017permutational}.
\cite{anderson2017some} and \cite{hamidi2019w}
correct the PERMANOVA statistic for heteroscedasticity in the unbalanced setting based on the variants of classical ANOVA by Brown--Forsythe and Welch, respectively.

In an independent line of research, \cite{dubey2019frechet} design an ANOVA procedure on metric spaces using Fr\'echet means as centroid objects. They propose to use as statistic the sum of an ANOVA-term and a Levene-term. We discuss this further in Subsection~\ref{subsec:DM}.

In the present paper we formulate Anderson's PERMANOVA on general metric spaces. We simply refer to the resulting method as Anderson ANOVA, because the use of M (due to the use of $\Rd$ in Anderson's work) seems inappropriate in our context and the use of PER (referring to the fact that a permutation test is performed) does not distinguish it from the other methods used. Rather than pursuing the PERMDISP method mentioned above, we introduce a new test for detecting differences of within-group dispersion based on Levene's procedure and refer to it as $L$-test. Our test statistic works directly with the pairwise distances between observations without using any kind of group centroid, neither in the original metric space nor in any principal coordinate space. We show that it has an asymptotic $\chi_1^2$-distribution, but we recommend using it with a permutation test just as the other statistics.

We also study the two summands used by \cite{dubey2019frechet} as separate test statistics for detecting differences in location and dispersion, respectively. We refer to Table~\ref{tab:DifferentStatistics} for an overview of the methods discussed.
\begin{table}[htb]
  \centering
        \begin{tabular}{|c|c|c|}
            \hline
            & \textit{location} & \textit{dispersion}
            \\
            \hline
            \textit{pairwise distances} & Anderson, Subsec.~\ref{subsec:Anderson} & New $L$-test, Section~\ref{sec:distanceLevene}
            \\
            \hline
            \textit{Fr\'echet means} & Dubey--Müller, Subsec.~\ref{subsec:DM} & Dubey--Müller, Subsec.~\ref{subsec:DM}
            \\
            \hline
        \end{tabular}
        \caption{\label{tab:DifferentStatistics} Overview of the non-Euclidean ANOVA methods studied in this paper. Procedures targeting \emph{location} are derived from the classic ANOVA statistic, whereas those targeting \emph{dispersion} are derived from the classic Levene's statistic (ANOVA statistic for ``deviations''). The rows distinguish whether computationally a procedure is based on (simple arithmetics of) \emph{pairwise distances} or on a centroid object (here a \emph{Fr\'echet mean}) in the metric space.}
\end{table}

Although the methods described are applicable in general metric spaces, our central goal in undertaking this research was to be able to perform ANOVA for point pattern data, see also the discussion section of \cite{muller2020metrics}. Among all the metric spaces mentioned above, we therefore focus in the later part of the present paper on the space of finite point patterns equipped with the TT-metric from~\cite{muller2020metrics}. As in many other spaces, exact Fr\'echet means can be computed within reasonable time only for (very) small data sets and one typically has to resort to a heuristic algorithm that finds only local minima of the Fr\'echet functional. We present simulation studies to compare the powers of the four tests across various situations and to understand the quality of approximation by the limiting $\chi_1^2$-distribution from a practical point of view. We also present an application of a 2-way ANOVA on a data set of bubbles in a mineral flotation process.

The plan of the paper is as follows: Section~2 contains a brief reminder of central aspects of classical ANOVA including Levene's test. In Section 3 we give a rather detailed presentation of Anderson ANOVA in metric spaces and the two summands proposed by Dubey and Müller. In Section 4 we introduce our new L-statistic, discuss its relation to the other methods and the original Levene's test, and show its asymptotic distribution. Section 5 is a short overview of the metric space of point patterns. In Sections~6 and 7 we present the simulation studies and the real-world data example. The paper ends with some further conclusions in Section 8.

\section{Classic ANOVA} \label{sec:classic}

For self-containedness and easy reference we briefly remind the reader of some facts and formulae in the context of the classical ANOVA going back to \cite{fisher1925statistical}. 
Details can be found in~\cite{scheffe1967analysis}.

\subsubsection*{One-Way ANOVA}

Given independent observations $x_{ij} \in \RR$, $1\leq j \leq n_i$, $1 \leq i \leq k$, from $k$ potentially different distributions $P_1,\ldots,P_k$, we do the following sum-of-squares decomposition
$$
  \tss = \mss + \rss,
$$
where
\begin{align*}
    && && \tss &= \sum_{i=1}^{k} \sum_{j=1}^{n_i} (x_{ij} - \bar{x}_{\cdot\cdot})^2 &&\text{\textit{(total sum of squares)}} && &&
    \\
    && && \mss &= \sum_{i=1}^{k}  n_i(\bar{x}_{i\cdot} - \bar{x}_{\cdot\cdot})^2 &&\text{\textit{(model sum of squares)}} && &&
    \\
    && && \rss &= \sum_{i=1}^{k} \sum_{j=1}^{n_i} (x_{ij} - \bar{x}_{i\cdot})^2 &&\text{\textit{(residual sum of squares).}} && &&
\end{align*}
Here $\bar{x}_{i\cdot} = \frac{1}{n_i}\sum_{j=1}^{n_i} x_{ij}$ denotes the $i$-th group mean and $\bar{x}_{\cdot\cdot} = \frac{1}{n}\sum_{i=1}^{k} \sum_{j=1}^{n_i} x_{ij}$ denotes the overall mean. We write $n = \sum_{i=1}^k n_i$ for the total number of observations.

Assume for now that the group distributions $P_i$ are Gaussian with the same variance. Under the null hypothesis that also the means are the same (hence all data comes from the same normal distribution), it is well-known that the ANOVA statistics
\begin{equation}  \label{eq:1wayF}
  F = \frac{n-k}{k-1}\frac{\mss}{\rss},
\end{equation}
describing the ratio between the variability explained by the model and the total variability in the data, is $F$-distributed with $k-1$ and $n-k$ degrees of freedom. Since $T \sim F(d_1,d_2)$ implies $d_1 T \inlawto \chi^2_{d_1}$ as $d_2 \to \infty$, we obtain
\begin{equation}
  (k-1) F \inlawto \chi^2_{k-1} \quad \text{as $n \to \infty$.}
\end{equation}
The \emph{asymptotic} result remains true even if $P_1 = P_2 = \ldots = P_k$ is non-Gaussian, but has second moments and there are $\lambda_1,\ldots,\lambda_k > 0$ such that the ratios of group sizes satisfy $\frac{n_i}{n} \to \lambda_i$, see e.g. \cite{wooldridge2010econometric}, Section 3.6.2.

\begin{remark} \label{rem:strictANOVA}
Strictly speaking ANOVA techniques are designed for inference within a linear model of different group means plus errors. Based on an error distribution $P$ with mean zero, one considers the model equations
\begin{equation*}
  x_{ij} = \mu_i + \eps_{ij}, \quad 1\leq j \leq n_i,\, 1 \leq i \leq k,
\end{equation*}
where $\mu_i \in \RR$ are the different group means and $\eps_{ij}$ are i.i.d.\ $P$-distributed error terms. In terms of the group distributions above this means that $P_i = P \ast \delta_{\mu_i}$, i.e.\ $P_i$ is obtained by shifting $P$ by $\mu_i$. Note that the asymptotic $\chi^2_{k-1}$-test does not need this assumption since in any case the null hypothesis just correspond to having $k$ times the same distribution. At the same time we cannot expect this test to achieve high power against all alternatives that have substantially different group distributions (see also the paragraph on Levene's test below). We will take up this point when discussing ANOVA on metric spaces, where typically ``shifting the distribution'' is meaningless (but may have an intuitive counterpart).
\end{remark}

\subsubsection*{Two-Way ANOVA}

As soon as more than one grouping factor is involved, important design decisions come into play, such as if factors are (partially) nested or if we allow for interaction terms between several factors on the same level. ANOVA has a long-standing history with many different designs. As an example which is pursued further in later sections we remind the reader of the balanced two-way ANOVA (two main factors, with interaction terms, same number $\tn$ of observations for each factor combination).

Given independent observations $x_{i_1 i_2 j} \in \RR$, $1\leq j \leq \tn$, $1 \leq i_1 \leq k_1$, $1 \leq i_2 \leq k_2$ from groups obtained by crossing a Factor $a$ with $k_1$ levels and a Factor $b$ with $k_2$ levels (with $n_{i_1 i_2} := \tn$ observations for each combination), we can perform a finer sum-of-squares decomposition
\begin{equation*}
  \tss = \ssa + \ssb + \ssi + \rss,
\end{equation*}
splitting up the model sum of squares into sums of squares for the individual factors and an interaction sum of squares. In formulae:
\begin{align*}
    \tss &= \sum_{i_1=1}^{k_1}\sum_{i_2=1}^{k_2}\sum_{j=1}^{\tn} (x_{i_1i_2j}-\bar{x}_{\cdot\cdot\cdot})^2
    \\
    \rss &= \sum_{i_1=1}^{k_1}\sum_{i_2=1}^{k_2}\sum_{j=1}^{\tn} (x_{i_1i_2j}-\bar{x}_{i_1i_2\cdot})^2
    \\
    \ssa &= \sum_{i_1=1}^{k_1} k_2 \tn (\bar{x}_{i_1\cdot\cdot}-\bar{x}_{\cdot\cdot\cdot})^2
    \\
    \ssb &= \sum_{i_2=1}^{k_2} k_1 \tn (\bar{x}_{\cdot i_2\cdot}-\bar{x}_{\cdot\cdot\cdot})^2
    \\
    \ssi &= \sum_{i_1=1}^{k_1}\sum_{i_2=1}^{k_2} \tn (\bar{x}_{i_1i_2\cdot} - \bar{x}_{i_1\cdot\cdot} - \bar{x}_{\cdot i_2\cdot} + \bar{x}_{\cdot\cdot\cdot})^2,
\end{align*}
where the various means are taken over the dot components while keeping the given indices fixed. Set $n = \sum_{i_1=1}^{k_1} \sum_{i_2=1}^{k_2} n_{i_1 i_2} = k_1 k_2 \tn$.

In addition to performing an omnibus test for group differences as for one-way ANOVA, we may then test for effects of Factor a and b separately, as well as for an interaction effect. 
The corresponding statistics are
\begin{align*}
    Fa = \frac{n - k_1 k_2}{k_1 - 1} \frac{\ssa}{\rss}, \quad
    Fb = \frac{n - k_1 k_2}{k_2 - 1} \frac{\ssb}{\rss}, \quad
    Fi = \frac{n - k_1 k_2}{(k_1-1)(k_2-1)} \frac{\ssi}{\rss}.
\end{align*}
If the observations come from Gaussian distributions with equal variances, each of the three statistics is $F$-distributed again under the corresponding null hypothesis that different levels of the factor or interaction to be tested do not lead to different shifts in mean. The degrees of freedom can be read from the denominator and the numerator, respectively, of the first ratio in each statistic.

\subsubsection*{Levene's Test}\label{subsec:Levene}

The test first proposed in \cite{levene1960robust}
was originally developed as a preliminary test to check for equal variances \emph{before} applying the basic ANOVA $F$-test in the Gaussian setting. This was important, as it was well-known at the time that for the goal of inference about differences in the means of the various groups (see Remark~\ref{rem:strictANOVA}), the size of the $F$-test can depart substantially from its nominal size if group variances are not equal.

\cite{levene1960robust} proposed to use as test statistic the usual ANOVA statistic, but to replace the observations $x_{ij}$ by the absolute differences from their group means $z_{ij} = \abs{x_{ij}-\bar{x}_{i\cdot}}$, i.e.
\begin{align}  \label{stat:Levene}
    \widetilde{F} = \frac{n-k}{k-1} \cdot \frac{\sum_{i=1}^k n_i (\bar{z}_{i\cdot}-\bar{z}_{\cdot\cdot})^2}
    {\sum_{i=1}^k \sum_{j=1}^{n_i} (z_{ij}-\bar{z}_{i\cdot})^2}.
\end{align}
If the observations are independently sampled from the same Gaussian distributions, it is plausible that $\widetilde{F}$ is still approximately $F$-distributed, because the dependence between the $z_{ij}$ is small even at moderate group sizes. 
This was confirmed by simulation in \cite{levene1960robust}. \cite{brown1974robust}
present a larger simulation experiment suggesting that replacing the $\bar{x}_{i\cdot}$ in the definition of $z_{ij}$ by a trimmed mean or median leads to a more robust test for non-Gaussian data. 

Current best practice suggests to perform a Welch-modified ANOVA directly if the as\-sump\-tion of equal variance is unclear as it results only in a small loss of power in the case where the variances are indeed equal. We refer to \cite{gastwirth2009impact} for a comprehensive presentation on Levene's test including this question and many further developments.

Levene's test and its variants remain highly important today as differences in variances (or some other measure of dispersion) are often in the center of attention in their own rights. In the rest of the paper we present tests on differences in ``location'' of groups and differences in ``dispersion'' of groups, both based on inter-point distances in a metric space. Our goal is to combine one of either kind in order to detect group differences in some universality.

\section{Non-Euclidean ANOVA}

In this and the next sections we assume that our data lies in a general metric space $(\mcx,d)$. We present existing methods of testing for group differences based on  ANOVA-like ideas. For the presentation we focus on generalizations of 1-way ANOVA, but provide further information on which methods can easily be extended to more complex designs. We always assume having $n = \sum_{i=1}^k n_i$ independent observations $x_{ij} \in \mcx$, $1\leq j \leq n_i$, $1 \leq i \leq k$ from $k$ potentially different distributions $P_1,\ldots,P_k$ on $\mcx$ (with Borel $\sigma$-algebra).

\subsection{Anderson ANOVA}\label{subsec:Anderson}

\cite{anderson2001new} argues, in the context of data sets in ecology, that traditional multivariate analogues of ANOVA are too stringent in their assumptions. These are typically based on similar statistics as \eqref{eq:1wayF}, but with absolute values replaced by Euclidean norms, see e.g. \cite{mardia1979multivariate}
Section 12.3. We may avoid the use of means of observations by writing $\tss - \rss$ instead of $\mss$ and replacing the sums of squared deviations from the mean with the help of the formula
\begin{align*}
  \sum_{j=1}^m \norm{y_j - \overline{y}}^2 = \frac{1}{2m}\sum_{j_1,j_2=1}^{m} \norm{y_{j_1}-y_{j_2}}^2 = \frac{1}{m}\sum_{j_1,j_2=1}^{m,\, <} \norm{y_{j_1}-y_{j_2}}^2,
\end{align*}
where we indicate by ``$<$'' in the summation bound that the sum is to be taken over strictly ordered summands only, here $j_1 < j_2$. Anderson proposes to replace the pairwise Euclidean distances by more general dissimilarities between observations and performs a permutation test. In our context we simply use the pairwise distances in the metric space. Thus
\begin{align*}
    \tss &= \frac{1}{n} \biggl( \sum_{i_1,i_2=1}^{k,\, <} \sum_{j_1=1}^{n_{i_1}} \sum_{j_2 = 1}^{n_{i_2}} d^2(x_{i_1 j_1}, x_{i_2 j_2}) + \sum_{i=1}^k \sum_{j_1,j_2=1}^{n_i,\,<} d^2(x_{i j_1}, x_{i j_2}) \biggr) \\
    \rss &= \sum_{i=1}^k \frac{1}{n_i} \sum_{j_1,j_2=1}^{n_i,\,<} d^2(x_{ij_1}, x_{ij_2}) 
    \\[1mm]
    \mss &= \tss-\rss
\end{align*}
and the final Anderson ANOVA statistic becomes
\begin{equation*}
      F_{\text{A}} = \frac{n-k}{k-1} \frac{\mss}{\rss}.
\end{equation*}

It has been noted in various places that this statistic may suffer from type I error inflation (in terms of a null hypothesis of equal \emph{means} in Euclidean space) and substantial loss of power in the unbalanced setting if the groups are heteroscedastic; see e.g.\cite{alekseyenko2016multivariate}. \cite{anderson2017some}
and \cite{hamidi2019w}
propose improvements based on the classical ANOVA variants by \cite{brown1974small} and \cite{welch1951comparison}, respectively. In the former, the $F$-statistic is replaced by
\begin{equation*}
  F_{\text{BF}} = \frac{\mss}{\sum_{i=1}^k (1-\frac{n_i}{n}) \frac{1}{n_i (n_i-1)} \sum_{j_1,j_2=1}^{n_i,\,<} d^2(x_{ij_1}, x_{ij_2})}.
\end{equation*}
For the simulation studies in Section~\ref{sec:simulstud} we concentrate on the balanced setting, for which Anderson $F_A$ performs typically well even in presence of heteroscedacity. We therefore do not discuss these improvements further, which in the balanced setting do not change the statistic.

\subsection{Fr\'echet ANOVA}\label{subsec:DM}

\cite{dubey2019frechet} introduce ANOVA-like terms that use distances in the metric $d$ to Fr\'echet means rather than absolute differences to averages. For observation $y_1, \ldots y_m \in \mcx$ the Fr\'echet mean is defined as
\begin{align} \label{eq:frechetmean}
    \bar{y} = \argmin_{z \in \mcx} \sum_{i=1}^m d^2(y_i,z).
\end{align}
One of the assumptions in \cite{dubey2019frechet} is that all Fr\'echet means considered exist and are unique. For our usual set of observations we denote by $\bar{x}_{i \cdot}$ the Fr\'echet mean of $x_{i1}, \ldots, x_{in_i}$, $i =1, \ldots, k$ and by $\bar{x}_{\cdot \cdot}$ the Fr\'echet mean of all observations. Following the notation in \cite{dubey2019frechet}, we write the Fr\'echet variance for the $i$-th group and the total Fr\'echet variance as
\begin{align*}
   \hat{V}_i = \frac{1}{n_i} \sum_{j=1}^{n_i} d^2(x_{ij}, \bar{x}_{i\cdot}) \quad \text{and} \quad \hat{V}_p = \frac{1}{n} \sum_{i=1}^k\sum_{j=1}^{n_i} d^2(x_{ij}, \bar{x}_{\cdot \cdot}),
\end{align*}
respectively. While $\hat{V}_i$ is the mean of $d^2(x_{ij}, \bar{x}_{i\cdot})$, $j =1, \ldots, n_i$, we also require the corresponding variance
\begin{align*}
  \hat{\sigma}_i^2 = \frac{1}{n_i} \sum_{j=1}^{n_i} d^4(x_{ij}, \bar{x}_{i\cdot}) - \hat{V}_i^2.
\end{align*}
Setting $\lambda_i = \frac{n_i}{n}$ one finally obtains
\begin{align*}
    U_n &= \sum_{i_1, i_2=1}^{k,\,<}\frac{\lambda_{i_1} \lambda_{i_2}}{\hat{\sigma}_{i_1}^2 \hat{\sigma}_{i_2}^2} (\hat{V}_{i_1} - \hat{V}_{i_2})^2
    \\
    F_n &= \hat{V}_p - \sum_{i=1}^k \lambda_i \hat{V}_i
    \\
    T &= \frac{nU_n}{\sum_{i=1}^k \frac{\lambda_i}{\hat{\sigma}_i^2}} + \frac{nF_n^2}{\sum_{i=1}^k \lambda_i^2 \hat{\sigma}_i^2} =: T_L + T_F.
\end{align*}

In the Euclidean setting of Section~\ref{sec:classic} the term $F_n$ is equal to $\frac{1}{n}(\tss-\rss)$ and the denominator of $T_F$ is then an estimator for the variance of $\frac{1}{n}\rss$, so that $T_F$ has close ties to the ANOVA F-statistic. The unweighted summands $(\hat{V}_{i_1} - \hat{V}_{i_2})^2$ of $U_n$ are similar in spirit to the terms $(\bar{z}_{i\cdot}-\bar{z}_{\cdot\cdot})^2$ from the definition of Levene's statistic, and in fact it appears that in the Euclidean case $T_L$ corresponds exactly to a simpler variant of Welch's ANOVA applied to $d^2(x_{ij}, \bar{x}_{i\cdot})$, $j = 1, \ldots, n_i$, $i = 1, \ldots, k$; see the computation in Formulae~(8)--(16) in \cite{hamidi2019w}. Thus $T_L$ has close ties to Levene's statistic.

Dubey and Müller show under a list of conditions pertaining to existence and uniqueness of theoretical and empirical Fr\'echet means and the complexity of the metric space (in terms of entropy integrals) that
\begin{equation*}
  \frac{nU_n}{\sum_{i=1}^k \frac{\lambda_i}{\hat{\sigma}_i^2}} \inlawto \chi^2_{k-1} \quad \text{and} \quad \frac{nF_n^2}{\sum_{i=1}^k \lambda_i^2 \hat{\sigma}_i^2} \inlawto 0 \quad \text{as $n \to \infty$}.
\end{equation*}
The authors advocate the simple addition of the two terms in order to obtain a single test statistic~$T$, maybe with weights if there is prior information available whether to rather look out for inequality of Fr\'echet means or of Fr\'echet variances. However, due to the unbalanced convergence of the two terms and the fact that the reason for the concrete normalization (especially) of $T_F$ remains a bit inscrutable to us, we prefer to analyze the two summands separately in Section~\ref{sec:simulstud}.

\section{A New Non-Euclidean Method of Levene Type}  \label{sec:distanceLevene}

What appears to be missing is a test for detecting differences of within-group dispersion that is based directly on pairwise distances between observations in the metric space. The idea of the PERMDISP-test mentioned in the introduction, i.e.\ performing multidimensional scaling and applying Levene's test in the principal coordinate space, is to some extent applicable here. However, it is rather an indirect method and it is methodologically not on the same level as the Anderson $F_A$. Indeed multidimensional scaling can be applied in combination with \emph{any} Euclidean procedure, so the PERMDISP-method should be rather paired up with the analog method of multidimensional scaling plus applying Euclidean (M)ANOVA. What is more, it contains an unwelcome tuning parameter, the number of principal coordinates, which is not easy to choose, but may be crucial.
Instead we propose the following test of Levene type for data in a metric space.

\subsection{Form and Properties} \label{subsec:distanceLeveneBasics}

We assume the same setup as in the previous section, i.e.\ there are $n = \sum_{i=1}^k n_i$ independent observations $x_{ij} \in \mcx$, $1\leq j \leq n_i$, $1 \leq i \leq k$ from $k$ potentially different distributions $P_1,\ldots,P_k$ on $\mcx$. Set
$N_i = \binom{n_i}{2}$ and $N = \sum_{i=1}^k N_i$. As a surrogate for the individual deviation terms $z_{ij}$ from Levene's statistic~\eqref{stat:Levene}, which in a general metric space would require the use of a Fr\'echet or similar mean, we use $d_{i, \{j_1,j_2\}} := \frac12 d(x_{ij_1},x_{ij_2})$. To simplify the notation, we enumerate the two-element subsets of $\{1,\ldots,n_i\}$ by $j = 1, \ldots, N_i$ and use $d_{ij}$ rather than $d_{i, \{j_1,j_2\}}$ for the $j$-th half-distance in the $i$-th group.

In a first step we assume that $n_1 = \ldots = n_k$ (balanced case) and emulate the statistics~\eqref{stat:Levene} by setting
\begin{align}\label{stat:distleveneBalanced}
  L &:= \frac{N-k}{k-1}\frac{\sum_{i=1}^{k}n_i(\bar{d}_{i\cdot}-\bar{d}_{\cdot\cdot})^2}{\sum_{i=1}^{k}\sum_{j=1}^{N_i}(d_{ij} - \bar{d}_{i\cdot})^2} 
\end{align}
where
\begin{align*}
\bar{d}_{i\cdot} = \frac{1}{N_i} \sum_{j=1}^{N_i} d_{ij} \quad \text{and} \quad
\bar{d}_{\cdot\cdot} = \frac{1}{N} \sum_{i=1}^k \sum_{j=1}^{N_i} d_{ij}
\end{align*}
denote the $i$-th group mean and the overall mean over pairwise distances, respectively.

Typographically the main fractions of Equations~\eqref{stat:distleveneBalanced} and \eqref{stat:Levene} are very similar, but the way they use the data $x_{ij}$ is quite different in that we replace $z_{ij} = \abs{x_{ij}-\bar{x}_{i\cdot}}$, 
$1 \leq j \leq n_i$ by $d_{i,\{j_1,j_2\}} = \frac12 d(x_{ij_1},x_{ij_2})$, $1 \leq j_1 < j_2 \leq n_i$. Note that we keep $n_i$ in the numerator rather than replacing it by $N_i$, which might have seemed more natural at first glance. The reason is the substantial dependence of the random variables $d_{i, \{j_1,j_2\}}$ (as opposed to the less substantial dependence between the $z_{ij}$) for each $i$, which implies that $n_i$, not $N_i$, is the correct scaling factor; see Subsection~\ref{subsec:limitdistr}. Note further that, for the same reason, the main denominator is not the most natural choice here, but it is convenient since it keeps the statistic similar to the original Levene statistic, 
is fast to compute and empirically performs no worse than the more natural choice discussed in Subsection~\ref{subsec:limitdistr}.

There are various ways how one might generalize \eqref{stat:distleveneBalanced} to general group sizes. We propose using
\begin{align}\label{stat:distlevene}
L &:= \frac{N-k}{k-1}\frac{\frac{1}{n}\sum_{i=1}^{k-1}\sum_{j=i+1}^{k}n_in_j(\bar{d}_{i\cdot}-\bar{d}_{j\cdot})^2}{\sum_{i=1}^{k}\sum_{j=1}^{N_i}(d_{ij} - \bar{d}_{i\cdot})^2}.
\end{align}
Direct computation shows that Equations~\eqref{stat:distlevene} and~\eqref{stat:distleveneBalanced} agree in the balanced case, but not in general; see Remark~\ref{rem:idemcomp}. The statistic~\eqref{stat:distlevene} performs well in several respects: it allows for an asymptotic distribution ($\chi_{k-1}^2$ up to a deterministic factor, see Corollary~\ref{cor:mainasymp}), is still fast to compute and shows a reasonable performance for unequal group sizes, though it may well be that a more judicious scaling that takes more proper care of different group sizes would be superior.

We briefly come back to this last point in Section~\ref{sec:simulstud}, but do not go much deeper in the present paper because based on additional considerations, both theoretical and from simulation studies, we do not see any clear improvements when choosing different normalizations.

In spite of the limit distribution which we compute in the next section, we recommend performing a permutation test as for the other statistics considered.
For this we permute the observations, not only their distances, i.e.\ new permutations use distances that are potentially different from the pairwise within-group distances of the original data. As a consequence not only the RSS changes with permutations, but also the TSS.

It is easy enough to generalize the construction of the above test statistic to more complex experimental designs. As an example we take up the balanced two-way ANOVA from Section~\ref{sec:classic} and form the corresponding Levene-type statistics for $(\mcx,d)$. For the specific statistics see Section \ref{subsec:balancedTwoWayLevene}.

\subsection{Limit Distribution}
\label{subsec:limitdistr}

In this subsection we derive asymptotic distributions for the statistic $L$ from~\eqref{stat:distlevene} and for the related statistic
\begin{align}  \label{stat:distlevenetilde}
  \widetilde{L} := \frac{N^*- k}{k-1} \frac{\frac{1}{n}\sum_{i=1}^{k-1}\sum_{j=i+1}^{k}n_in_j(\bar{d}_{i\cdot}-\bar{d}_{j\cdot})^2}{4\,T_n},
\end{align}
where $N^* = \sum_{i=1}^k n_i(n_i-1)^2$ and
\begin{align}  \label{eq:hatgamma_ultimate}
  T_n = \sum_{i=1}^k
  \sum_{\substack{j_1,j_2,j_3=1 \\ j_1 \not\in\{j_2,j_3\}}}^{n_i} \bigl(d_{i,\{j_1,j_2\}} - \bar{d}_{i\cdot}\bigr) \bigl( d_{i,\{j_1,j_3\}} - \bar{d}_{i\cdot}\bigr). 
\end{align}
The previous formula makes it necessary to use the more complicated notation $d_{i,\{j_1,j_2\}} = \frac12 d(x_{ij_1},x_{ij_2})$ from the beginning of Subsection~\ref{subsec:distanceLeveneBasics}. Note that $\frac{1}{N^*-k} T_n$ is a natural group based estimator of $\mathrm{Cov}\bigl(\frac12 d(X_1,X_2), \frac12 d(X_1,X_3) \bigr)$, where $X_1,X_2,X_3$ are three independent random variables sampled from the distribution of the group. The normalization by ${N^*-k}$ rather than ${N^{*}}$ is simply modeled after the bias correcting term for independent data points.

In spite of the ANOVA-like construction, we cannot use the asymptotic theory for ANOVA directly,
because the distances $d_{i,\{j_1,j_2\}}$, our ``data'', stem from dependent random variables for each $i$. This dependence is taken into account by using the factor $\frac{n_i n_j}{n}$ rather than $N_i$ or $N_j$ in the numerator and by normalizing with $\frac{1}{N^*-k} 4\, T_n$ in~\eqref{stat:distlevenetilde}, which then still allows to obtain the asymptotic $\chi^2_{k-1}$-distribution for $(k-1) \widetilde{L}$. In contrast $(k-1) L$ converges ``only'' towards a multiple of $\chi^2_{k-1}$ that depends on parameters of the group distribution. 
\begin{theo} \label{thm:mainasymp}
   Assume that the Borel $\sigma$-algebra for $(\mcx,d)$ is countably generated.
   In the usual 1-way setup of Subsection~\ref{subsec:distanceLeveneBasics} assume that $P_1= \ldots = P_k = P$ for a distribution $P$ that is not a Dirac distribution and satisfies $\int_{\mcx} \int_{\mcx} d^2(x,y) \, P(dx) \, P(dy) < \infty$. Suppose that there are $\lambda_i > 0$ such that $n_i/n \to \lambda_i$ for every $i$ as $n \to \infty$. Then we have
   \begin{equation*}
     (k-1) \, \widetilde{L} \, \inlawto \, \chi^2_{k-1} \quad \text{as $n \to \infty$.}
   \end{equation*}
\end{theo}
\begin{corollary} \label{cor:mainasymp}
  Under the conditions of Theorem~\ref{thm:mainasymp}, we obtain
  \begin{equation*}
     (k-1) \, L \, \inlawto \, \frac{4 \gamma^2}{\sigma^2} \chi^2_{k-1} \quad \text{as $n \to \infty$},
   \end{equation*}
   where with independent $X,Y,Z \sim P$ we have
   \begin{align*}
     \gamma^2 &= \mathrm{Cov}(d(X,Y), d(X,Z)); \\
     \sigma^2 &= \mathrm{Var}(d(X,Y)).
   \end{align*}
\end{corollary}
\begin{proof}[Proof of Theorem~\ref{thm:mainasymp}]
  Under the null hypothesis our data is generated by independent $\mcx$-valued random elements $X_{ij} \sim P$, $1\leq j \leq n_i$, $1 \leq i \leq k$, and the distances $d_{i,\{j_1,j_2\}}$ are realizations of the random variables $\frac{1}{2}d(X_{ij_1},X_{ij_2})$, $1\leq j_1 < j_2 \leq n_i$, $1 \leq i \leq k$. Under the conditions on $P$ we have asymptotic normality of the $U$-statistics
 \begin{align}
   U_{i} = U_i^{(n)} = {\binom{n_i}{2}}^{-1} \sum_{j_1,j_2=1}^{n_i,<} \tfrac{1}{2}d(X_{ij_1},X_{ij_2}), \quad i = 1, \ldots, k
 \end{align}
by a straightforward generalization of Hoeffding's theorem to random elements in $\mcx$, see The\-o\-rem~\ref{theo:hoeffding} in the appendix. More precisely, we have with $X,Y,Z \sim P$ independent that
\begin{equation}  \label{eq:hoeffding_an_uni}
   \sqrt{n_i} \bigl( U_{i} - \tfrac12 \mathbb{E} d(X,Y) \bigr) \inlawto \mathcal{N}(0,\gamma^2) \quad \text{as $n_i \to \infty$},
\end{equation}
where $\gamma^2 = \mathrm{Cov}(d(X,Y), d(X,Z)) = \mathrm{Var}(\mathbb{E} (d(X,Y) \, \vert \, X)) = 4\gamma_{h}^2$ in the notation of the appendix with $h = \frac12 d$. 
In view of the 1-way ANOVA construction, on which $L$ is based, we define the ``design matrix'' $D = D_n \in \RR^{n\times k}$ by
\begin{align}\label{matrixD}
    D^\prime := \begin{pmatrix}
        1 & \ldots & 1 & 0 & \cdots & 0 & \cdots & 0 & \cdots & 0
        \\
        0 & \cdots & 0 & 1 & \ldots & 1 & \cdots & 0 & \cdots & 0
        \\
         &  \vdots & & & & & \ddots & & \vdots &
        \\
        0 & \cdots & 0 & 0 & \cdots & 0 & \cdots & 1 & \ldots & 1
    \end{pmatrix} \in \RR^{k\times n},
\end{align}
where the $i$-th row has exactly $n_i$ ones, and the ``contrast matrix''
\begin{align}\label{matrixC}
    C := \begin{pmatrix}
        1 & 0  & \cdots & 0 & -1
        \\
        0 & 1  &   & 0 & -1
        \\
        \vdots &  & \ddots &  & \vdots
        \\
        0 & 0 &  & 1 & -1
    \end{pmatrix} \in \mathbb{R}^{(k-1)\times k}.
\end{align}
Setting $\Delta = \lim_{n \to \infty} \frac{1}{n} D_n'D_n = \mathrm{diag}(\lambda_1,\ldots,\lambda_n)$, we obtain with $U = U^{(n)} = (U_1, \ldots, U_k)'$  by independence of the components and $n_i \to \infty$ as $n \to \infty$ (since $\lambda_i > 0$) that
\begin{equation}  \label{eq:hoeffding_an_mult}
   Z_n := \gamma^{-1} \sqrt{n} \, \Delta^{1/2} \bigl( U - \EE U \bigr) \inlawto \mathcal{N}_k(0,I_k) \quad \text{as $n \to \infty$}.
\end{equation}
Setting $\nu = (n_1, \ldots, n_k)'$, we may further compute $C'(C(D'D)^{-1}C')^{-1} C = D'D - \frac{1}{n} \nu \nu'$ (see Lemma~\ref{lem:idemcomp} in the Appendix for the calculation) and therefore
\begin{equation}
  \widetilde{L} = \frac{N^*- k}{k-1} \frac{U' C'(C(D_n'D_n)^{-1}C')^{-1} C U}{4 \, T_n}.
\end{equation}
Since $\EE U = \frac{1}{2} \mathbb{E} d(X,Y) \cdot \mathds{1} \in \RR^k$ and $C \cdot \mathds{1} = 0$, we obtain
\begin{equation*}
  (k-1) \widetilde{L} = \gamma^2 \frac{Z_n' \bigl(\tfrac{1}{n}W_n \bigr) Z_n}{\frac{4}{N^*-k} T_n},
\end{equation*}
where $W_n := \Delta^{-1/2} C'(C(D_n'D_n)^{-1}C')^{-1} C \Delta^{-1/2}$. Note that $$W := \lim_{n \to \infty} \frac{1}{n} W_n = \Delta^{-1/2} C'(C \Delta^{-1} C')^{-1} C \Delta^{-1/2}$$ is a symmetric and idempotent matrix of rank $k-1$, and therefore $Z' W Z \sim \chi^2_{k-1}$ for $Z \sim \mcn_k(0,I_k)$ by Lemma~\ref{lem:cochran} from the Appendix. Using \eqref{eq:hoeffding_an_mult} it is straightforward to show with the help of the continuous mapping theorem that
\begin{equation*}
  Z_n' \bigl(\tfrac{1}{n}W_n \bigr) Z_n \inlawto \chi^2_{k-1}.
\end{equation*}

So it suffices to show that $\frac{1}{N^*-k} T_n \inpto \gamma_{d/2}^2$.
For this we note that the normalized inner sum of \eqref{eq:hatgamma_ultimate} satisfies
\begin{align}
  &\frac{1}{n_i(n_i-1)^2}
  \sum_{\substack{j_1,j_2,j_3=1 \\ j_1 \not\in\{j_2,j_3\}}}^{n_i} \bigl(d_{i,\{j_1,j_2\}} - \bar{d}_{i\cdot}\bigr) \bigl( d_{i,\{j_1,j_3\}} - \bar{d}_{i\cdot}\bigr) \notag\\
  &\hspace*{6mm} = \underbrace{\frac{n_i(n_i-1)(n_i-2)}{n_i(n_i-1)^2}}_{\longrightarrow \, 1} \underbrace{\frac{1}{n_i(n_i-1)(n_i-2)}
  \sum_{j_1,j_2,j_3=1}^{n_i,\, \neq} \bigl(d_{i,\{j_1,j_2\}} - \bar{d}_{i\cdot}\bigr) \bigl( d_{i,\{j_1,j_3\}} - \bar{d}_{i\cdot}\bigr)}_{\longrightarrow \, \mathrm{Cov}(\frac12 d(X_1,X_2), \frac12 d(X_1,X_3)) \, = \, \gamma^2_{d/2}} \label{eq:consistent}\\
  &\hspace*{12mm} {} + \underbrace{\frac{1}{n_i-1}}_{\longrightarrow \, 0} \underbrace{\frac{1}{n_i(n_i-1)} \sum_{j_1,j_2}^{n_i,\, \neq} \bigl(d_{i,\{j_1,j_2\}} - \bar{d}_{i\cdot}\bigr)^2}_{\longrightarrow \, \mathrm{Var}(\frac12 d(X_1,X_2)) = \sigma^2_{d/2}}, \notag
\end{align}
where convergence of the averages is almost surely and follows after expansion of the products by the strong law of large numbers for $U$-statistics using  the prerequisite $\EE(d(X_1,X_2)^2) < \infty$; see \cite{hoeffding1961strong}.

Thus for the total term
\begin{align*}
  \frac{1}{N^*-k} T_n &= \frac{1}{N^*-k}\sum_{i=1}^k n_i(n_i-1)^2\cdot\frac{1}{n_i(n_i-1)^2}
  \sum_{\substack{j_1,j_2,j_3=1 \\ j_1 \not\in\{j_2,j_3\}}}^{n_i} \bigl(d_{i,\{j_1,j_2\}} - \bar{d}_{i\cdot}\bigr) \bigl( d_{i,\{j_1,j_3\}} - \bar{d}_{i\cdot}\bigr) \longrightarrow \gamma^2_{d/2}.
\end{align*}
\end{proof}

\begin{proof}[Proof of Corollary~\ref{cor:mainasymp}]
  This follows from Theorem~\ref{thm:mainasymp} because
  \begin{equation*}
    L = 4 \frac{\frac{1}{N^*-k} T_n}{\frac{1}{N-k} \sum_{i=1}^{k}\sum_{j=1}^{N_i}(d_{ij} - \bar{d}_{i\cdot})^2} \widetilde{L},
  \end{equation*}
  where the numerator is a consistent estimator of $\gamma^2/4$ and the denominator is a consistent estimator of $\sigma^2/4$, see~\eqref{eq:consistent}.
\end{proof}

\section{Metric Space of Finite Point Patterns}
\label{sec:pp-space}

In Sections \ref{sec:simulstud} and \ref{sec:realdata} we apply the four statistics from Table~\ref{tab:DifferentStatistics} for the space of finite point patterns equipped with the metric introduced in~\cite{muller2020metrics}. For self-containedness we give a short summary of the relevant concepts and results, referring to the paper as MSM20.

For $n \in \ZZ_+$ write $[n] = \{1,2,\ldots , n\}$ (including $[0] = \emptyset$).
Denote by $\mfnfin$ the space of finite multisets on a complete separable metric space $(\mcr,\varrho)$. We refer to the elements $\xi = \lbrace x_1, x_2, \ldots , x_n \rbrace \in \mfnfin$ as point patterns, where $n \in \ZZ_{+} = \{0,1,2,\ldots\}$ and $x_i \in \mathcal{X}$, $i \in [n]$. Note that $x_i = x_j$ for $i \neq j$ is allowed and that the point patterns can be identified with the counting measure $\sum_{i=1}^n \delta_{x_i}$, which is often helpful for theoretical considerations. We write $\abs{\xi}$ to denote the total number of points in the pattern $\xi$.

\begin{de}[Definition 1 of~MSM20]
Let $C > 0$ and $p \geq 1$ be two parameters, referred to as \emph{penalty} and \emph{order}, respectively.

\nin
For $\xi = \lbrace x_1, \ldots , x_m \rbrace, \eta = \lbrace y_1, \ldots , y_n \rbrace \in \mfnfin$ define the \emph{transport-transform (TT) metric} by
\begin{equation}
    \label{eq:ttdef}
    \begin{split}
    \dtt(\xi,\eta) = \dtt^{(C,p)}(\xi,\eta)
    =\biggl( \min \biggl( (m+n-2l) C^p +\sum_{r=1}^{l} \varrho(x_{i_r},y_{j_r})^p \biggr) \biggr)^{1/p},
    \end{split}
\end{equation}
where the minimum is taken over equal numbers of pairwise different indices $i_1,\ldots,i_l$ in $[m]$ and $j_1,\ldots,j_l$ in $[n]$, i.e.\ over the set
\begin{equation*}
    \begin{split}
        S(m,n) = \bigl\{ (i_1,&\ldots,i_l;\hspace*{1.5pt} j_1,\ldots,j_l)\,;\; l \in \{0,1,\ldots,\min\{m,n\}\},\\
        &i_1, \ldots, i_l \in [m] \text{ pairwise different},\, j_1, \ldots, j_l \in [n] \text{ pairwise different} \bigr\}.
    \end{split}
\end{equation*}
\end{de}

The distance $\dtt(\xi,\eta)$ can be computed by filling up the smaller point pattern with dummy points located at distance $C$ 
until it has the same cardinality $n$ as the larger point pattern and then solving a standard assignment problem with cost $\min\{d(x,y), 2^{1/p}C\}$ between points $x, y$ (MSM20, Theorem~1). The classical worst-case complexity of this is $O(n^3)$ (MSM20, Remark~1), which can be somewhat improved to order $n^{2.5}$ up to polylogarithmic factors (\citealp{lee2014path}). Practical computation times for well over $n=1000$ points are less than one second (R package \textsf{ttbary}, \citealp{ttbary}, using the auction algorithm from \citealp{bertsekas1988auction}).

The TT-metric can be interpreted as an unbalanced Wasserstein metric (Remark~3). 
Com\-put\-ing Fr\'echet means in Wasserstein spaces is a topic of active research; see e.g.\ \cite{borgwardt2020improved}, \cite{borgwardt2021computational}, \cite{heinemann2021} and references therein for recent developments the space of discrete measures. In our context an additional increase in difficulty comes from the constraint that the result must be a discrete measure with integer cardinality. In~MSM20 we therefore apply an alternating heuristics to obtain local minima of the Fr\'echet functional in~\eqref{eq:frechetmean}. The resulting ``pseudo-barycenters'' are obtained much faster and appear to be of good quality (consistent objective function values and results conform with intuition), but are by no means perfect and still require considerable computation time for hundreds of patterns with hundred of points (Table~1--4 in~MSM20).

A related metric that we take up in Section~\ref{sec:realdata} is the relative TT-metric defined as
\begin{equation}
    \label{eq:rttdef}
      \drtt(\xi,\eta) = \drtt^{(C,p)}(\xi,\eta) = \frac{1}{\max\{\abs{\xi},\abs{\eta}\}^{1/p}} \dtt^{(C,p)}(\xi,\eta).
\end{equation}
This metric is in a sense more robust to individual outliers if there are many points. In particular note that $\drtt(\xi_N,\xi_N\cup\zeta) \to 0$ as $N \to \infty$ if $\abs{\xi_N} \to \infty$ and $\zeta$ is a fixed point pattern.

In view of the conditions for Theorem~\ref{thm:mainasymp}, completeness and separability are inherited from $(\mcx,\varrho)$ to $(\mfnfin,\dtt)$ and $(\mfnfin,\drtt)$. This is straightforward to see after checking that $\drtt(\xi_N,\xi) \to 0$ iff $\dtt(\xi_N,\xi) \to 0$ iff $\abs{\xi_N} \to \abs{\xi}$ and each point $x$ of $\xi$ is approximated by exactly one point of $\xi_N$ (if $x$ is a multipoint of cardinality $k$ this means that there is a total of exactly $k$ points in $\xi_N$, possibly forming multipoints of their own, that converge towards $x$). The condition $\int_{\mcx} \int_{\mcx} d(x,y) \, P(dx) \, P(dy) < \infty$ is always satisfied for $\drtt$ because it is bounded by $C$. Since $\dtt(\xi,\eta) \leq C \max\{\abs{\xi},\abs{\eta}\}^{1/p}$ it is satisfied for $\dtt$ if $\Xi \sim P$ satisfies $\EE\abs{\Xi}^{2/p} < \infty$, 
which is for example the case for all point process distributions considered in Section~\ref{sec:simulstud}.

For the simulation study in the next section it is helpful to understand some basic probability measures on $\mfnfin$. Suppose that $\mcr \subset \Rd$ is compact (in the next section we only use a unit square in $\RR^2$). A random element in the metric space $(\mfnfin,\drtt)$, equipped with its Borel $\sigma$-algebra is called a \emph{point process}, i.e.\ a point process is a measurable map from a probability space to $\mfnfin$. The Borel $\sigma$-algebra coincides with the smallest $\sigma$-algebra that makes $\xi \mapsto \xi(A)$ measurable for every measurable $A \subset \mcr$, which is the usual $\sigma$-algebra considered on $\mfnfin$; see Proposition~9.1.IV in~\cite{dvj2}.

We say a point process $\Xi$ satisfies \emph{complete spatial randomness (CSR)} if it is a Poisson process with intensity measure $\nu = \lambda \, \Leb^d$, where $\lambda \geq 0$ and $\Leb^d$ is Lebesgue measure (on $\mcr$). This means that $\Xi(A) \sim \mathrm{Po}(\nu(A))$ for all measurable $A \subset \mcr$ and that $\Xi(A_1),\ldots,\Xi(A_l)$ are independent for all $l \in \NN$ and all measurable $A_1, \ldots, A_l \subset \mcr$ that are pairwise disjoint. See e.g.\ Section~2.4 in~\cite{dvj1} for more details on the Poisson process.

\section{Simulation Study} \label{sec:simulstud}

We tested the different statistics from Table~\ref{tab:DifferentStatistics} for various point process distributions and present the results in what follows. 
First we investigate the practical use of our asymptotics in Subsection~\ref{subsec:simulstud_asymp}. In spatial statistics there are usually two fundamentally different ways how dis\-tri\-bu\-tions can deviate from CSR. 
One is spatial inhomogeneity of points, i.e.\ points may be more or less likely to occur in different regions of the space. The ability of tests to detect deviations from CSR against various spatially inhomogeneous alternatives is studied in Subsection~\ref{subsec:simulstud_inhomo}. The other way is interaction of points, i.e.\ presence of points in one region may excite or inhibit the presence of other points nearby. In Subsection~\ref{subsec:simulstud_interact} we study how well the statistics discern between various interaction strengths in homogeneous Strauss processes.

For the evaluations in Subsections~\ref{subsec:simulstud_inhomo} and~\ref{subsec:simulstud_interact} we perform permutation tests. These are based on generating $M$ independent uniform permutations of the indices of the data points resulting in alternative split-ups of the data into $k$ groups of sizes $n_i$, $1 \leq i \leq k$. We then determine the rank $r$ of the statistic-value for the original split-up within the statistic-values of the alternative split up (from $r = 1$ for the highest value to $r = M+1$ for the lowest value). It is easily checked (and well-known) that $p = \frac{r}{M+1}$ is an honest p-value (i.e.\ $\mathbb{P}(p \leq \alpha) \leq \alpha$ for every $\alpha \in (0,1)$). We reject the null if $p \leq 0.05$.

In Subsections~\ref{subsec:simulstud_inhomo} and~\ref{subsec:simulstud_interact} we have $k=2$, $n_i = \tn = 20$ and use $M = 999$ permutations if no barycenter computation is needed and $M = 99$ permutations if barycenter computation is needed. In view of the $\binom{40}{20} \approx 1.4 \cdot 10^{11}$ possible split-ups, this means that there is a high degree of randomization in each individual test. The small $M=99$ was necessary due to the large computational burden of computing pseudo-barycenters in point pattern space (see Section~\ref{sec:pp-space}). For statistics that do not require barycenter computation, choosing $M=999$ typically results in much faster computation time than the choice of $M=99$ for statistics that do require barycenter computation. For reproducibility of individual test results, a higher $M$ or (where possible) comparing within all possible split-ups into groups would be desirable in both cases.

Preferring exact permutation tests over tests based on the limit $\chi^2$-distribution is in agree\-ment with the recommendations from previous papers and corresponds to our own experience. 
However, the $\chi^2$-approximation of our $L$-statistic is quite fast as we can see in Subsection~\ref{subsec:simulstud_asymp}, where we compare the finite sample distributions of the new $L$- and the Dubey--Müller statistics.

In all tests we use as the underlying space $\mcr=[0,1]^2 \subset \RR^2$ with the Euclidean metric. The significance level is always $\alpha=0.05$. Furthermore we choose as order $p=2$ and as penalty $C=0.25$, which means that $\sqrt{2} \cdot 0.25 \approx 0.35$ is the maximal contribution that a single matched point pair makes to the TT-distance, i.e.\ the actual Euclidean distances are cut off at this value. In applications the choice of $C$ is often based on the physical reality of the data and possibly the goal of the analysis. For the present simulation study we tried not to restrict a substantial proportion of matching distances while keeping the contribution of additional points reasonably low. Table~\ref{tab:CompareCutoffsPois} gives an overview of how many pairs are matched above and below the cutoff distance 
for various values of $C$ based on pairwise comparisons of 1000 point patterns simulated according to CSR with intensity $\lambda=35$. For $C=0.25$ we have for every matching above the cutoff distance 
$1/0.038 \approx 26$ matchings below the cutoff distance.

\begin{table}[!htb]
    \begin{center}
        \begin{tabular}{|c||C{5.5em}|C{5.5em}|C{5.5em}|C{5.5em}|}
            \hline
               & mean $\dtt$ & below cutoff & above cutoff & unpaired \\
            \hline
            \hline
            $C=0.1$ & 0.309 & 11123388 & 4469722 & 3309250\\
            \hline
            relative &  & 1 & 0.424 & 0.311 \\
            \hline
            \hline
            $C=0.15$ & 0.393 & 13456877 & 2136233 & 3309250 \\
            \hline
            relative &  & 1 & 0.167 & 0.257 \\
            \hline
            \hline
            $C=0.2$ & 0.457 & 14514420 & 1078690 & 3309250 \\
            \hline
            relative &  & 1 & 0.078 & 0.24 \\
            \hline
            \hline
            $C=0.25$ & 0.512 & 15054688 & 538422 & 3309250 \\
            \hline
            relative &  & 1 & 0.038 & 0.233 \\
            \hline
            \hline
            $C=0.3$ & 0.561 & 15347529 & 245581 & 3309250 \\
            \hline
            relative &  & 1 & 0.017 & 0.23 \\
            \hline
            \hline
            $C=0.35$ & 0.609 & 15498175 & 94935 & 3309250 \\
            \hline
            relative &  & 1 & 0.006 & 0.229 \\
            \hline
        \end{tabular}
    \end{center}
    \vspace{-4mm}

    \caption{\label{tab:CompareCutoffsPois} Pairwise comparison within $1000$ patterns simulated from CSR on $[0,1]^2$ with intensity $\lambda=35$ for various penalties $C$. The columns give the mean $\dtt$-distance, and the (absolute and relative) number of matchings below cutoff and above cutoff, as well as the number of unpaired points for these $\binom{1000}{2}$ pairwise comparisons. Note that the relative numbers are with respect to the number of matchings below cutoff.}
\end{table}

\subsection{Asymptotics}
\label{subsec:simulstud_asymp}

In the present subsection we numerically assess the speed of convergence of our new $\widetilde{L}$ statistic under the null hypothesis of equal group distributions towards the $\chi^2_{k-1}$ distribution as presented in Subsection~\ref{subsec:limitdistr}. For comparison we also consider the Fr\'echet $T_L$ and $T$ statistics, which were shown in \cite{dubey2019frechet} to have a limiting $\chi^2_{k-1}$ distribution as well.

Our experiments are based on $k=2$ groups, both simulated from the same distribution, which is either $\csr(35)$ or the Strauss hard core distribution with $\lambda=35$. These are the extreme distributions having either no interaction or very strong interaction in Subsection~\ref{subsec:simulstud_interact}. As group size we consider $\tn = 5,20,50,200$.
The computation of the Fr\'echet $T$ and $T_L$ depend on the calculation of a barycenter. For this we used the heuristic algorithm presented in \cite{muller2020metrics}. The calculation of an exact barycenter is computationally infeasible for this kind of data. To compensate that we do not get the optimal solution, we did $5$ restarts in every barycenter calculation and used the best of the $5$ solutions as the barycenter. 

Figure~\ref{fig:AsymptoticsResults_CSR} shows QQ-plots for the empirical distributions of our new Levene statistic $\widetilde{L}$, the Fr\'echet statistic $T_L$ and the Fr\'echet statistic $T$ on the $y$-axis and the theoretical $\chi_1^2$ distribution on the $x$-axis.
The data are the $\csr(35)$ point patterns. In the first column the groups consist of $\tilde{n} = 5$ patterns, in the second column of $\tilde{n} = 20$ patterns and so on. 
For the two Levene statistics $\widetilde{L}$ and $T_L$ we can see the computed quantiles approach the theoretical quantiles as the group size $\tilde{n}$ gets larger.
Even for a medium group size $\tilde{n}=50$ the computed quantiles are very close to the theoretical quantiles of a $\chi_1^2$ distribution.

Similarly Figure~\ref{fig:AsymptoticsResults_Str} shows QQ-plots for hardcore Strauss distributed point patterns. Again the four columns correspond to the four group sizes $\tilde{n}=5,20,50,200$ and the three rows correspond to the three statistics.
For this data the computed quantiles are already very close to the theoretical quantiles of a $\chi_1^2$ distribution for $\tilde{n} =20$ for the two Levene statistics.

In both cases the third row, the combined Fr\'echet statistic $T$, yields quantiles that are far from the theoretical quantiles. This is solely due to the summand $T_F$ that is not considered in the second row.

In spite of the asymptotic results we use permutation based tests in what follows. 
This is in the tradition of previous methods, see e.g. \cite{anderson2017permutational}, \cite{dubey2019frechet}.
Comparison of different statistics for different data sets are presented in the following two sub\-sec\-tions.

\begin{figure}[!htb]
  \begin{center}
    \includegraphics[width=0.98\textwidth]{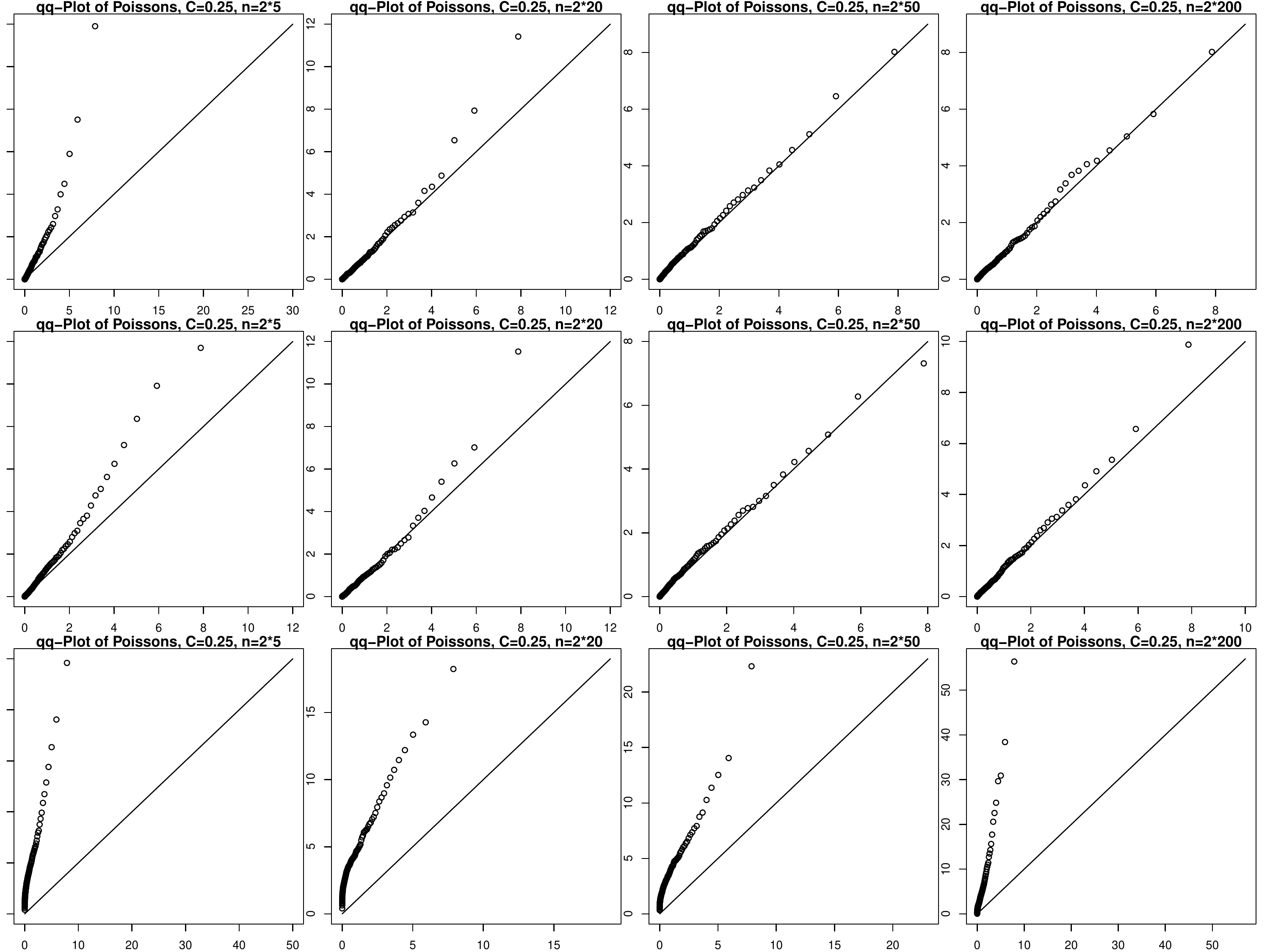}
  \end{center}
  \vspace*{-4mm}
     \caption{QQ-plots of the percentiles based on 500 statistics values (on the $y$-axis) versus $\chi^2_1$-percentiles. Based on $k=2$ groups of $\tn = 5,20,50,200$ patterns from $\csr(35)$. 
     The first row is our new $\widetilde{L}$ statistic \eqref{stat:distlevenetilde}, the second and third rows are the Fr\'echet $T_L$ statistic from Section \ref{subsec:DM} and the Fr\'echet $T$ statistic, respectively. }
    \label{fig:AsymptoticsResults_CSR}
\end{figure}
\begin{figure}[!htb]
  \begin{center}
    \includegraphics[width=0.98\textwidth]{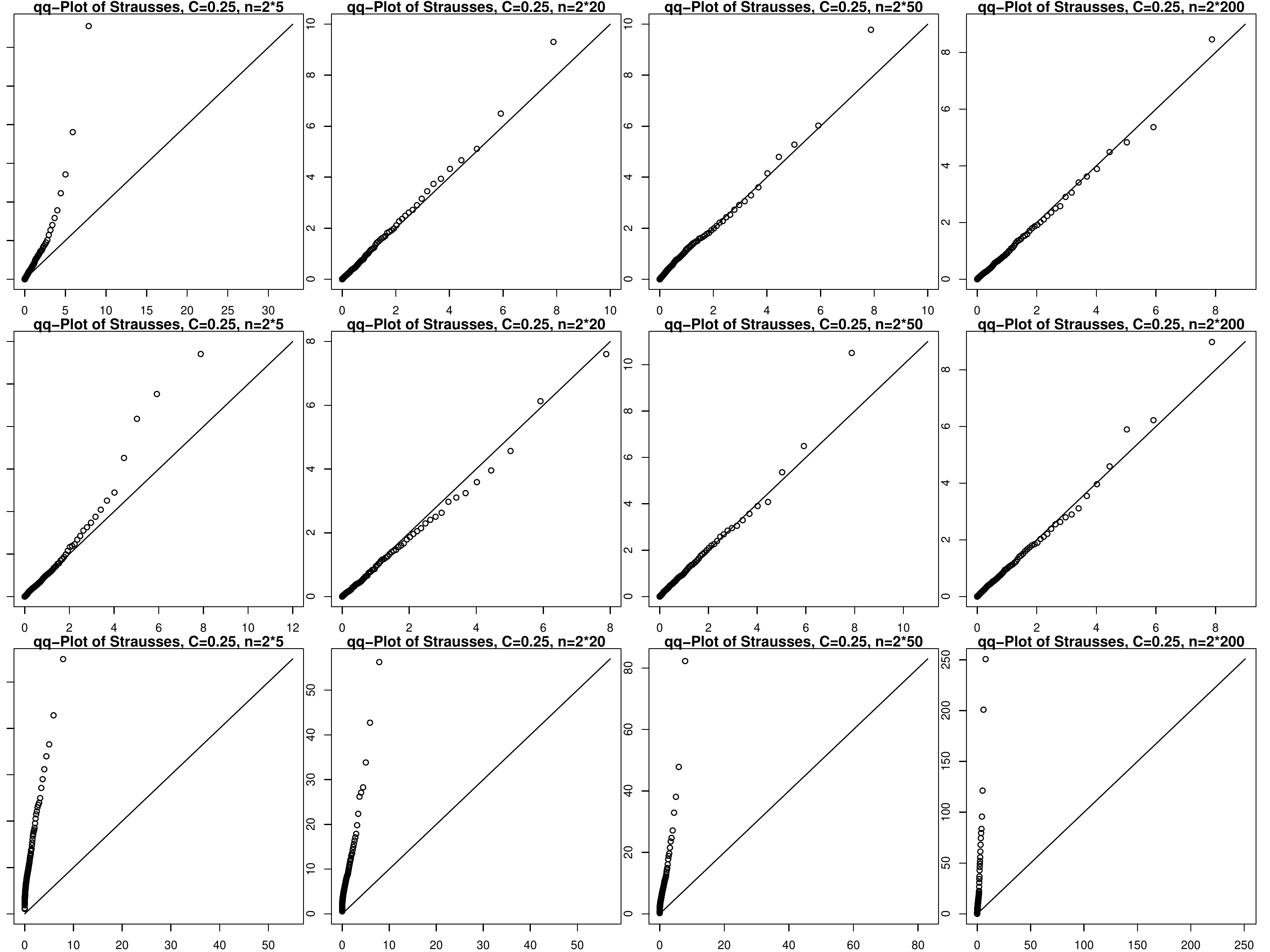}
  \end{center}
  \vspace*{-4mm}
     \caption{QQ-plots of the percentiles based on 500 statistics values (on the $y$-axis) versus $\chi^2_1$-percentiles. Based on $k=2$ groups of $\tn = 5,20,50,200$ patterns from 
     a Strauss hard core distribution with $\lambda=35$. 
     The first row is our new $\widetilde{L}$ statistic \eqref{stat:distlevenetilde}, the second and third rows are the Fr\'echet $T_L$ statistic from Section \ref{subsec:DM} and the Fr\'echet $T$ statistic, respectively.}
    \label{fig:AsymptoticsResults_Str}
\end{figure}

\subsection{Inhomogeneity}
\label{subsec:simulstud_inhomo}

Here we compare $k=2$ groups of $\tn = 20$ point patterns.
Patterns in Group~2 are simulated from CSR with $\lambda = 35$. In Group 1, they are simulated from various inhomogeneous scenarios, i.e.\ from Poisson process distributions where the intensity function (the density of the measure $\nu$ with respect to Lebesgue measure) deviates more or less from a constant but still integrates up to $35$ over the whole window $\mcr=[0,1]^2$.

In Scenarios 1--3 the intensity is obtained by adding a number of rotation-invariant Gaussian distributions with different means but the same covariance matrix $\sigma^2 I$ and scaling to total mass~$35$. For simplicity we do not restrict the intensity to $\mcr$, but as can be seen from Figure~\ref{fig:InhomoScenarios} only very few points outside $\mcr$ occur. Scenarios 4--6 use as intensity an exponential function that is constant in the $y$-coordinate and induces a certain tendency for points to lie in the left part of the window rather than in the right part.

Table~\ref{tab:InhomoScenarios} provides more information about the chosen parameters. Figure~\ref{fig:InhomoScenarios} shows five example point patterns for each scenario. In addition we add a Scenario~0, which corresponds to simulating the first group also from CSR with $\lambda = 35$.
\begin{table}[!htb]
  \renewcommand{\arraystretch}{1.3}
    \begin{center}
        \begin{tabular}{|c||l|}
            \hline
            Scenario & $\lambda(x,y)$ proportional to\\
            \hline
            $1$ & $\sum_{i=1}^3 \varphi_{\mu_i, 0.075}(x,y)$\\
            \hline
            $2$ & $\sum_{i=1}^3 \varphi_{\mu_i, 0.1}(x,y)$\\
            \hline
            $3$ & $\sum_{i=1}^4 \varphi_{\mu_i, 0.1}(x,y)$\\
            \hline
            $4$ & $\exp(-2x)$ \\ 
            \hline
            $5$ & $\exp(-1x)$\\ 
            \hline
            $6$ & $\exp(-0.02x)$\\ 
            \hline
        \end{tabular}
    \end{center}
      \vspace{-4mm}

    \caption{\label{tab:InhomoScenarios} Overview of the Poisson process intensities for the six scenarios. The proportionality constant is chosen such that the expected number of points in each scenario is $35$. By $\varphi_{\mu,\sigma^2}$ we denote the density of the bivariate normal distribution with mean $\mu \in \RR^2$ and covariance matrix~$\sigma^2 I$. The different $\mu_i$ used are seen in Figure~\ref{fig:InhomoScenarios}.}
\end{table}

\begin{figure}[!htb]
    \hspace*{-1em}
    \includegraphics[width=0.4\textwidth]{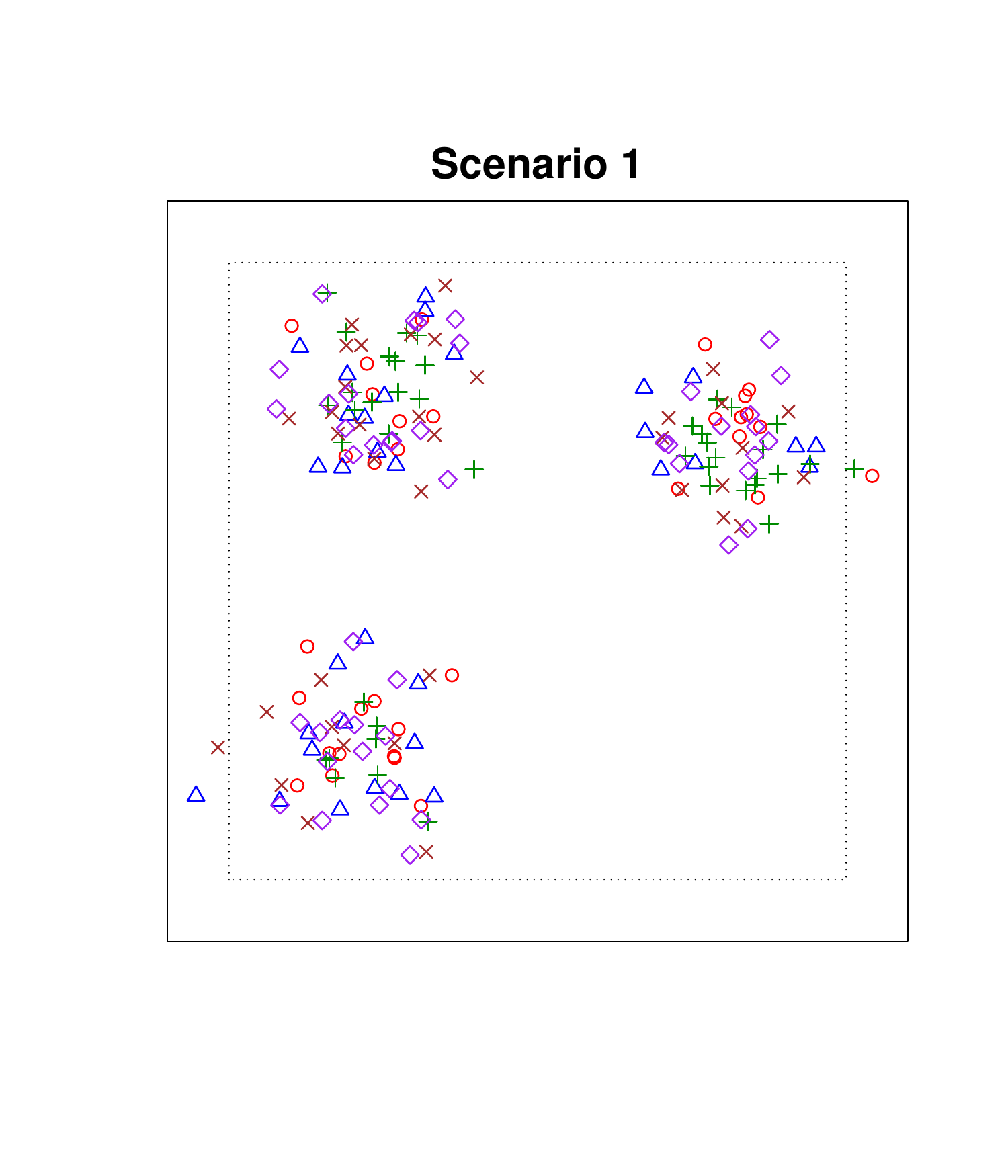}\hspace*{-4em}
    \includegraphics[width=0.4\textwidth]{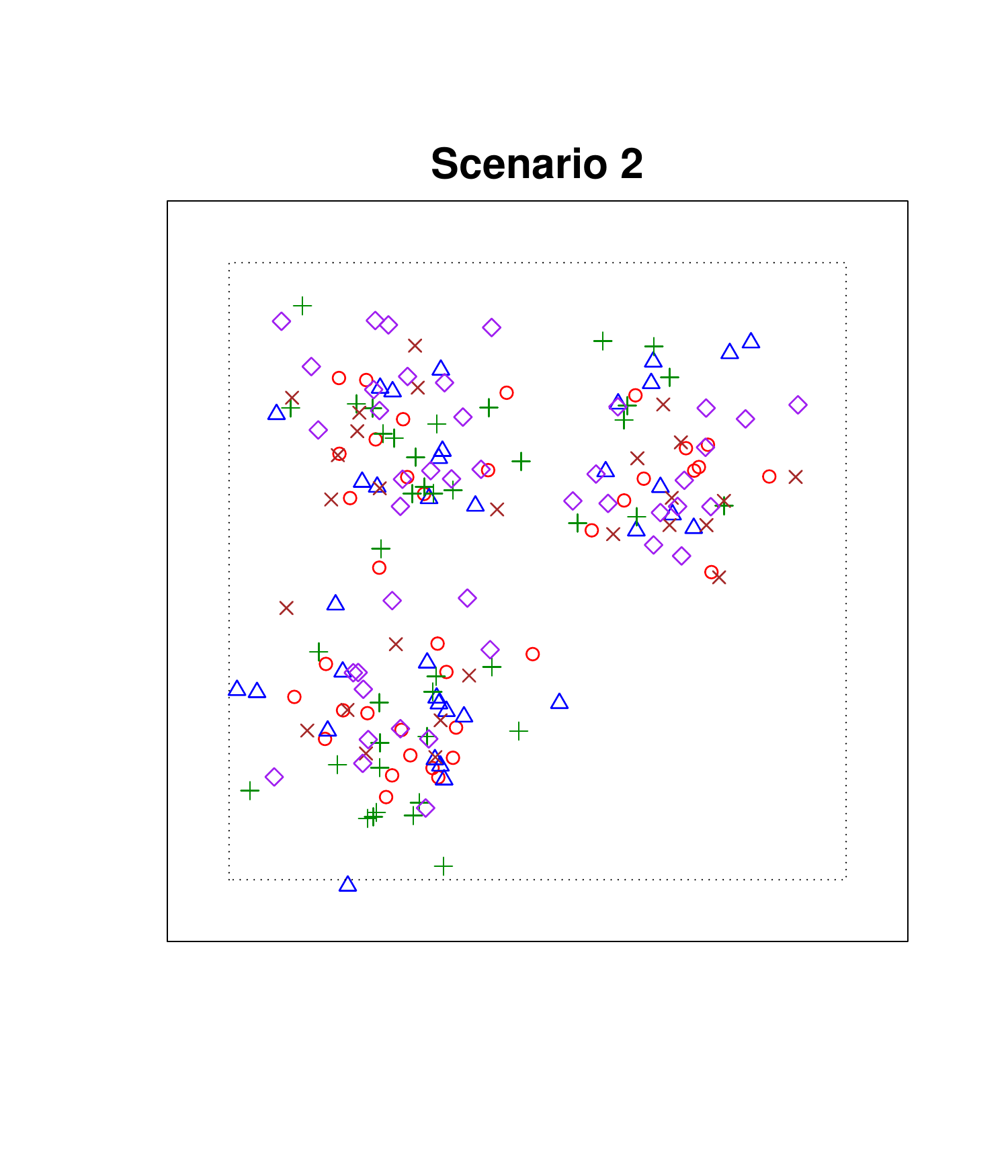}\hspace*{-4em}
    \includegraphics[width=0.4\textwidth]{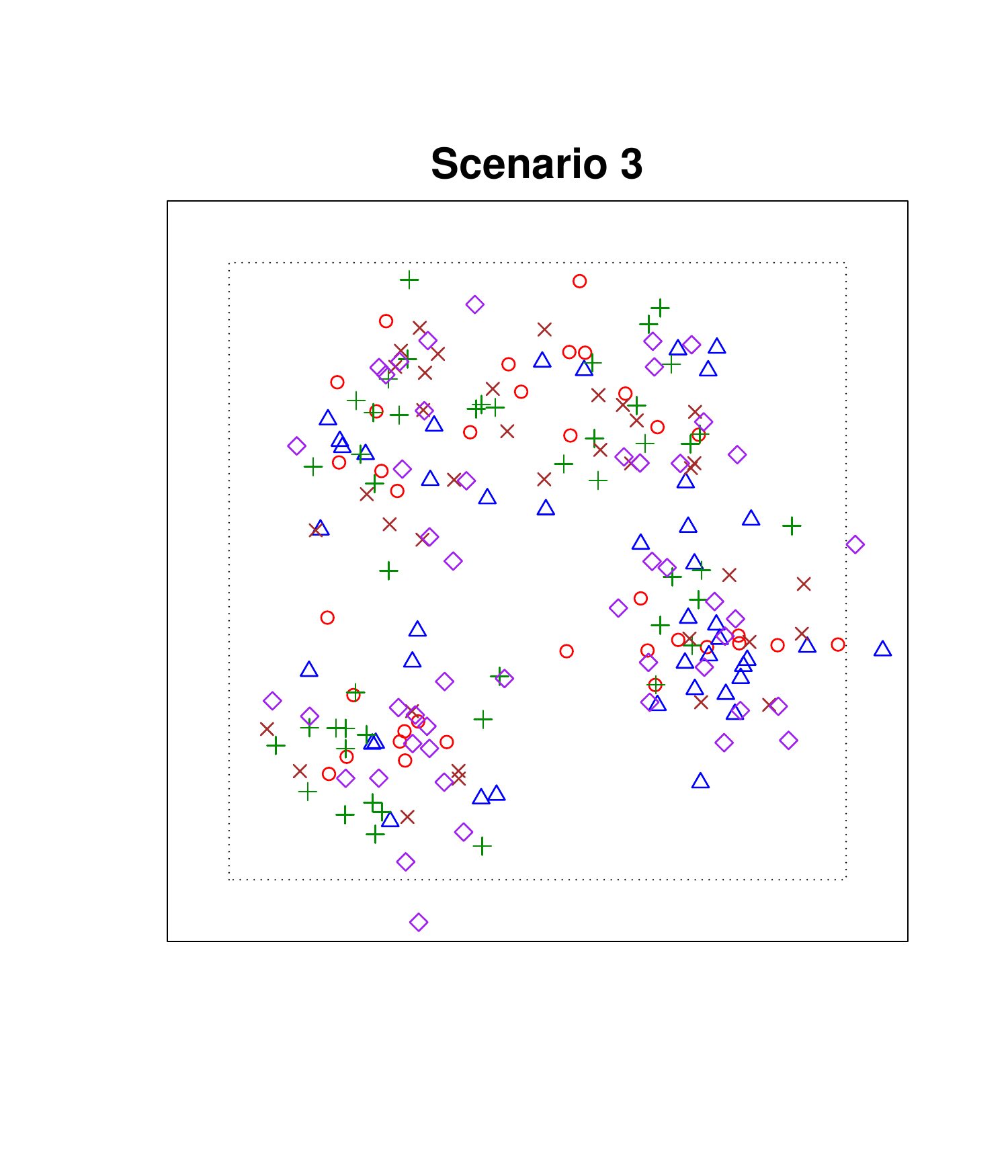}\vspace*{-4em}
    \hspace*{-1em}
    \includegraphics[width=0.4\textwidth]{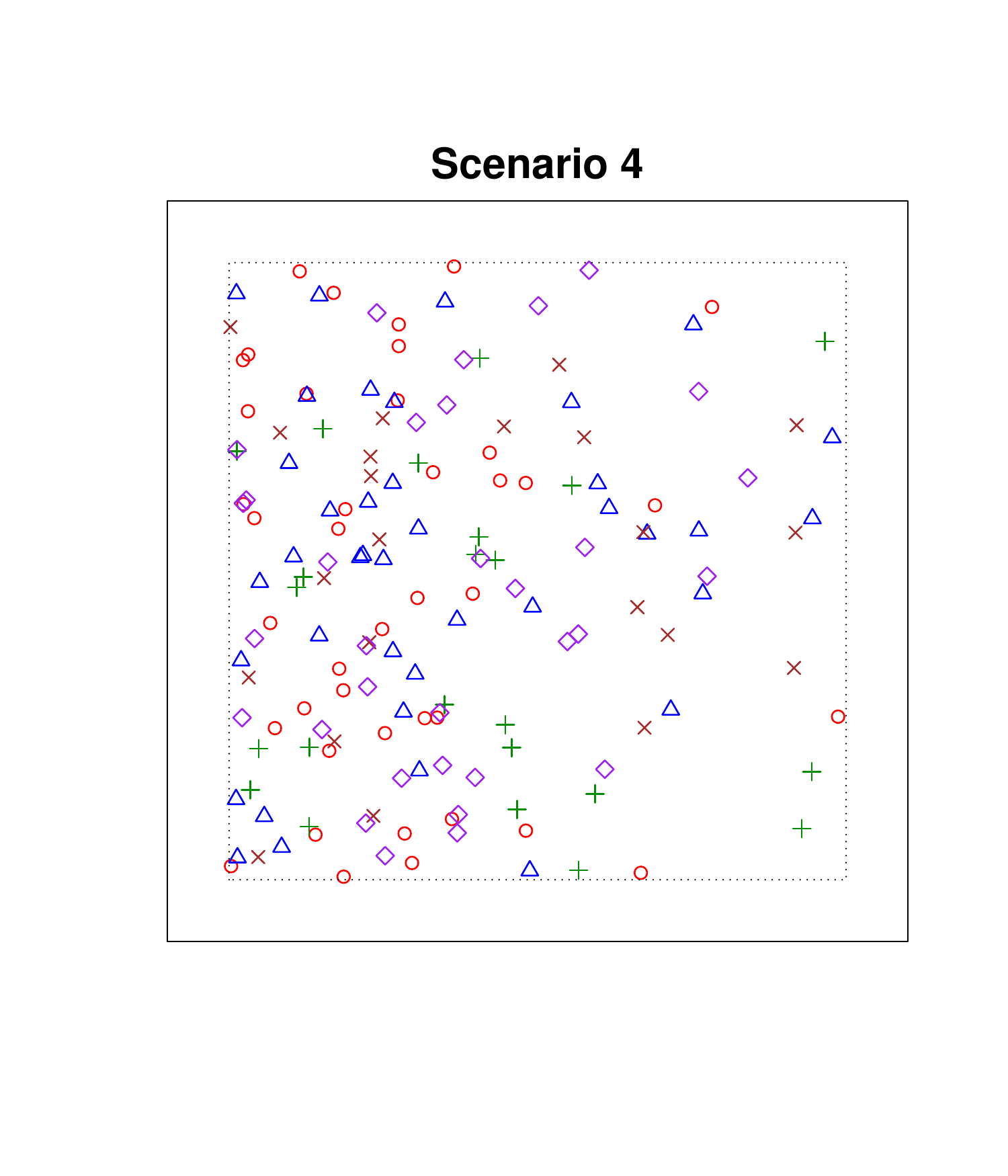}\hspace*{-4em}
    \includegraphics[width=0.4\textwidth]{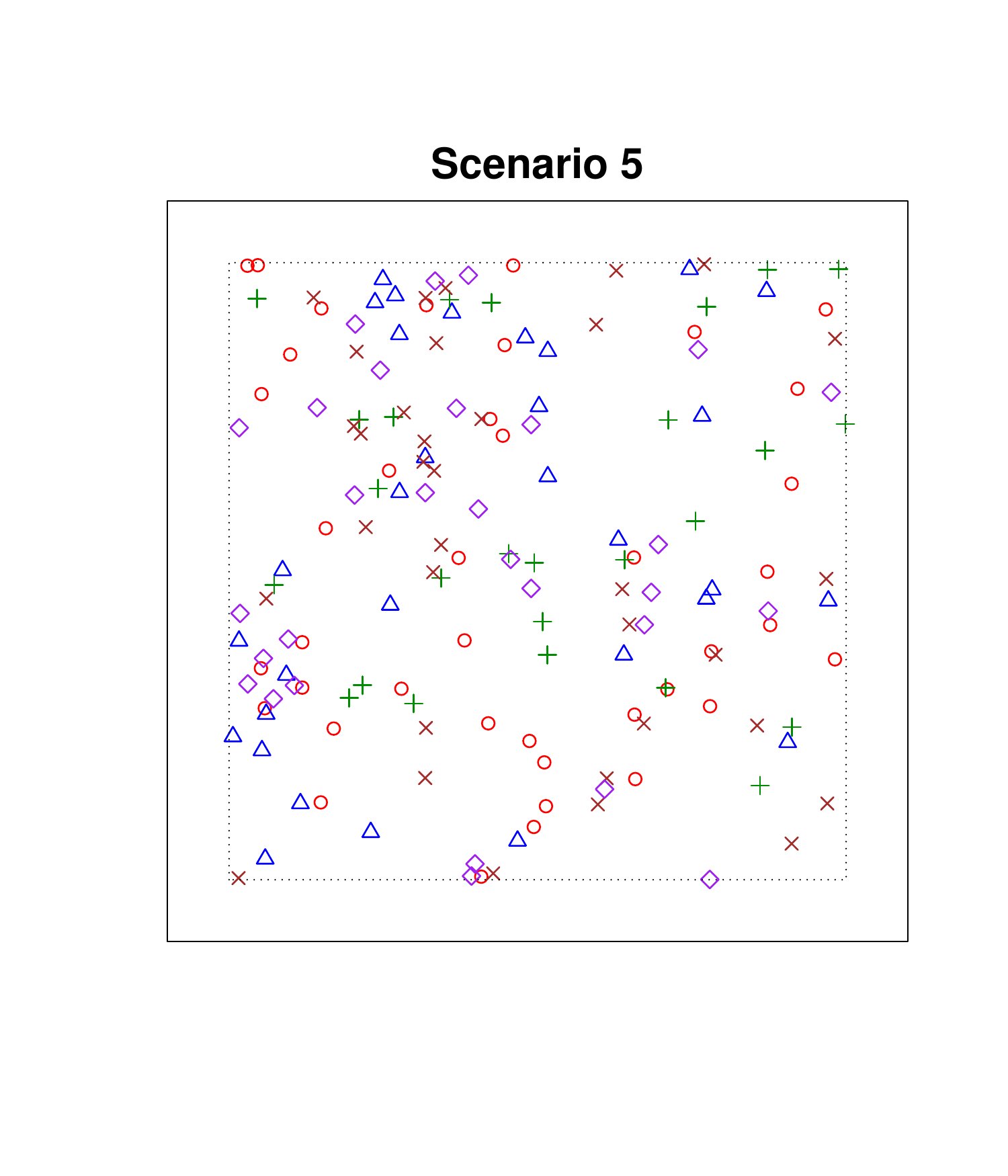}\hspace*{-4em}
    \includegraphics[width=0.4\textwidth]{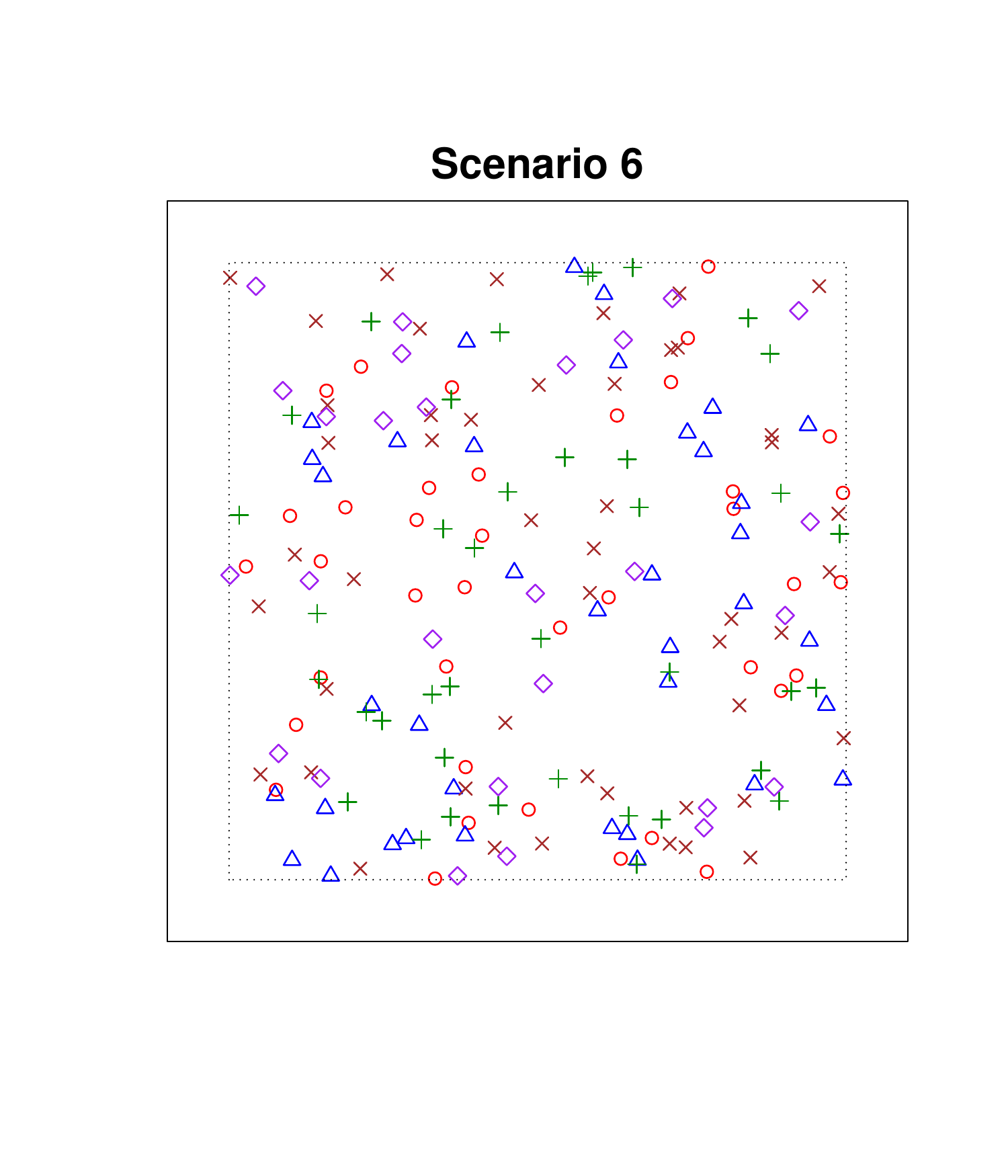}\vspace*{-4em}
    \caption{\label{fig:InhomoScenarios} Five example patterns for each of the six scenarios together with the window $[0,1]^2$ (dotted line) in which
    the homogeneous Poissons processes of the second group are sampled.}
\end{figure}

Table~\ref{tab:InhomoResults} gives the results in terms of numbers of rejections (out of 100) of the null hypothesis of equal distribution in both groups.
\begin{table}[htb]
    \begin{center}
        \begin{tabular}{|c||C{1.5em}|C{1.5em}|C{1.5em}|C{1.5em}|C{1.5em}|C{1.5em}|C{1.5em}|}
            \hline
            Scenario & $1$ & $2$ & $3$ & $4$ & $5$ & $6$ & $0$ \\
            \hline
            \hline
            Anderson $F_A$ & 100 & 100 & 100 & 100 & 99 & 39 & 2\\
            \hline
            new $L$ & 93 & 76 & 77  & 14 & 7 & 9 & 3 \\
            \hline
            Fr\'echet $T_F$ & 100 & 100 & 100 & 99 & 11 & 0 & 4 \\
            \hline
            Fr\'echet $T_L$ & 59 & 24 & 47 & 13 & 9 & 12 & 4 \\
            \hline
        \end{tabular}
    \end{center}
    \vspace*{-4mm}
    \caption{\label{tab:InhomoResults} Numbers of rejections of the null hypothesis ``equal distribution in both groups'' based on 100 data sets per column. In each data set the first group is sampled from the scenario indicated in the column and the second group is sampled from Scenario 0.}
\end{table}

We observe that the direct ANOVA procedures perform much better than the Levene (or indirect ANOVA) procedures. This is not so surprising, because the inhomogeneity experiment considers two groups of distributions that are different in terms of their location in the point pattern space. To see this intuitively, think about the distributions in Scenarios 1--6 (and 0 as a boundary case) in terms of producing locally perturbed versions of a typical point pattern, which is more or less any one of the example point patterns in Figure~\ref{fig:InhomoScenarios} (more appropriately one would rather think of an idealized version of these patterns, such as the Fr\'echet mean). Among the direct ANOVA methods, Anderson $F_A$ performs substantially better than Fr\'echet $T_F$ and has still a reasonable chance to detect the faint differences between Scenarios 6 and 0 when presented with the 20 patterns from each group. Our new L-test performs somewhat better than the Fr\'echet L-test, but both tests are only able to detect the inhomogeneity (with reasonable probability) when it is very obvious (Scenarios 1--3).

\subsection{Interaction between Points}
\label{subsec:simulstud_interact}

Again we compare $k=2$ groups of $\tn = 20$ point patterns. This time the group distributions differ in the degree of point interaction. For this we consider the distribution of the homogeneous Strauss process on the unit square $\mcr = [0,1]^2$, which is obtained by specifying the density $f \colon \mfnfin \to \RR_+$,
\begin{equation*}
  f(\xi) := c \cdot \beta^{\abs{\xi}} \cdot \gamma^{s_R(\xi)},
\end{equation*}
with respect to CSR with intensity 1 on $\mcr$, where
\begin{equation*}
  s_R(\xi) = \sum_{\{x,y\} \subset \xi} \mathds{1}\{\norm{x-y} \leq R\}
\end{equation*}
is the number of pairs of points at distance $\leq R$ from one another. Here $R>0$ is the range of the interaction, $\beta > 0$ is the so-called activity (which controls the intensity of the process via an increasing function, that is however only accessible numerically) and $\gamma \in [0,1]$ is the strength of the interaction. The constant $c$ normalizes the density to an overall integral of $1$ and is also not available in closed form. We write $\strauss(\beta, \gamma; R)$ for this point process distribution. Intuitively a $\strauss(\beta, \gamma; R)$-process is obtained from a $\csr(\beta)$ process by penalizing each outcome according to a factor $\gamma$ per $R$-close point pair. Correspondingly we have $\strauss(\beta, 1; R) = \csr(\beta)$ (regardless of $R$). At the other end of the spectrum $\strauss(\beta, 0; R)$ is the distribution of a hard core process with no points allowed within distance $R$ of other points.

For the simulation we set $R=0.1$ and consider scenarios based on the six different values $\gamma = 0,0.2,0.4,0.6,0.8,1$. The activity $\beta$ is adapted so that each time $\lambda=35$. Figure~\ref{fig:StraussExamples} shows one realization for each of the six scenarios.
\begin{figure}[!htb]
    \hspace*{2mm}
    \vspace*{-8mm}

    \hspace*{-1em}\includegraphics[width=0.42\textwidth]{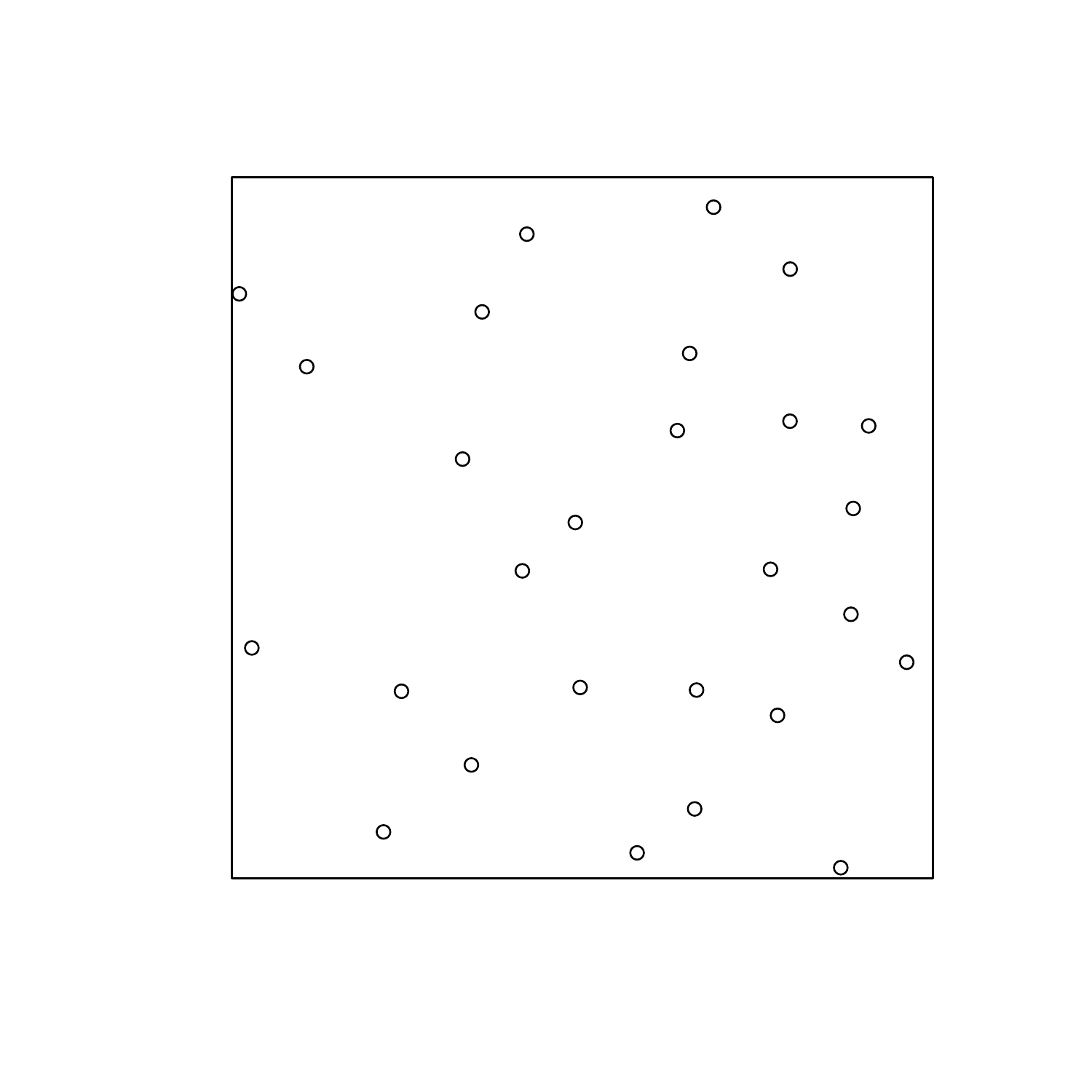}\hspace*{-5.5em}
    \includegraphics[width=0.42\textwidth]{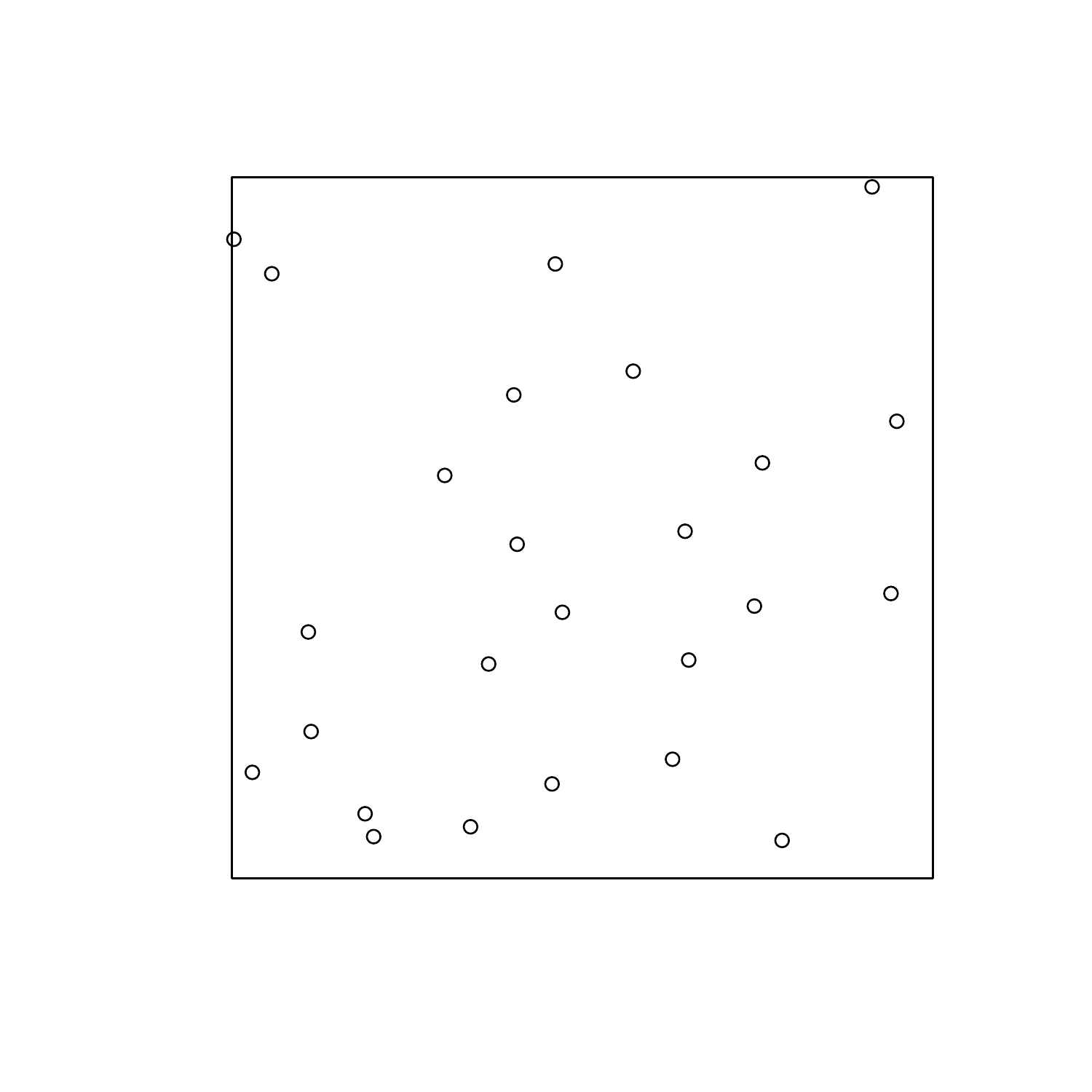}\hspace*{-5.5em}
    \includegraphics[width=0.42\textwidth]{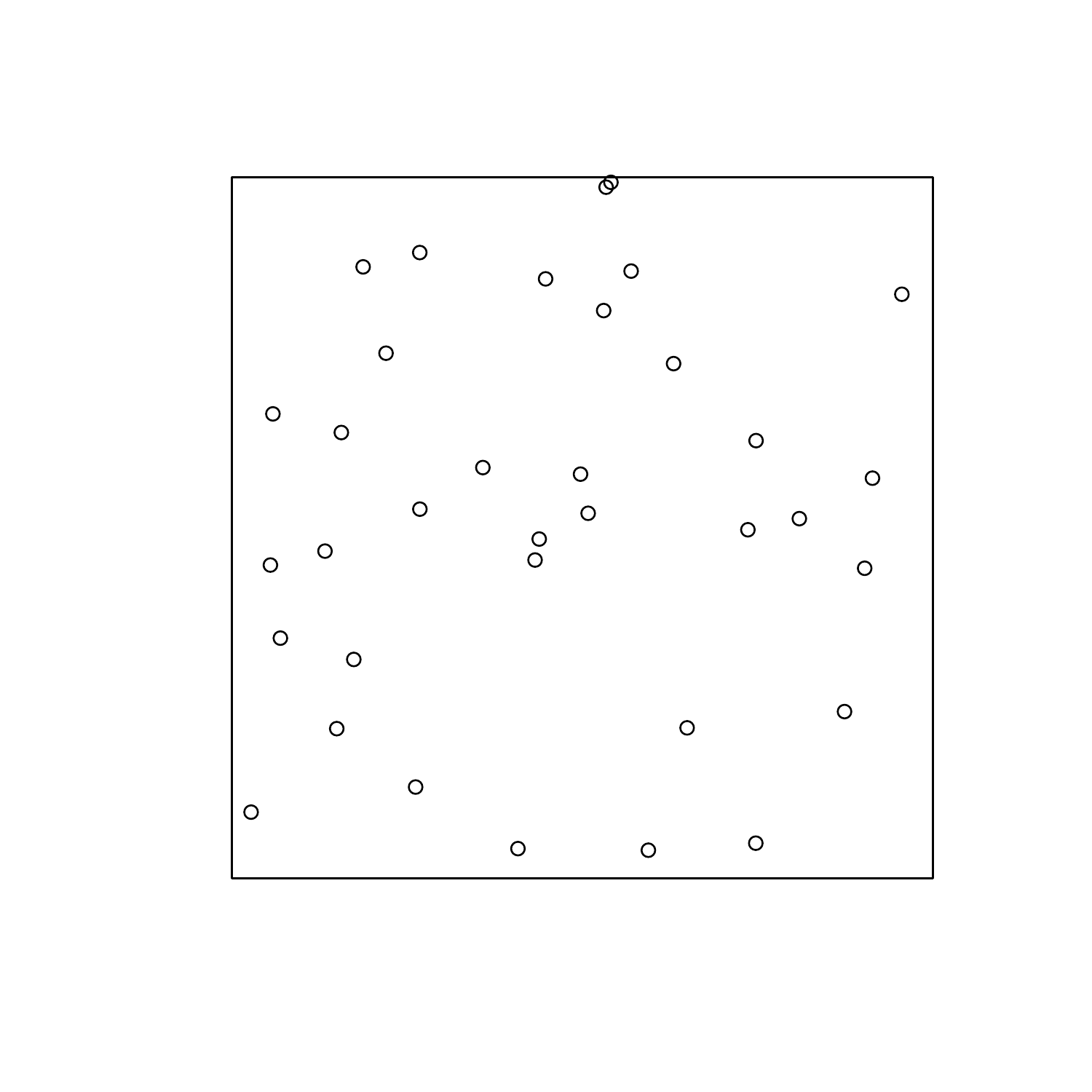}\hspace*{-5.5em}
    \vspace*{-20.5mm}

    \hspace*{-1em}\includegraphics[width=0.42\textwidth]{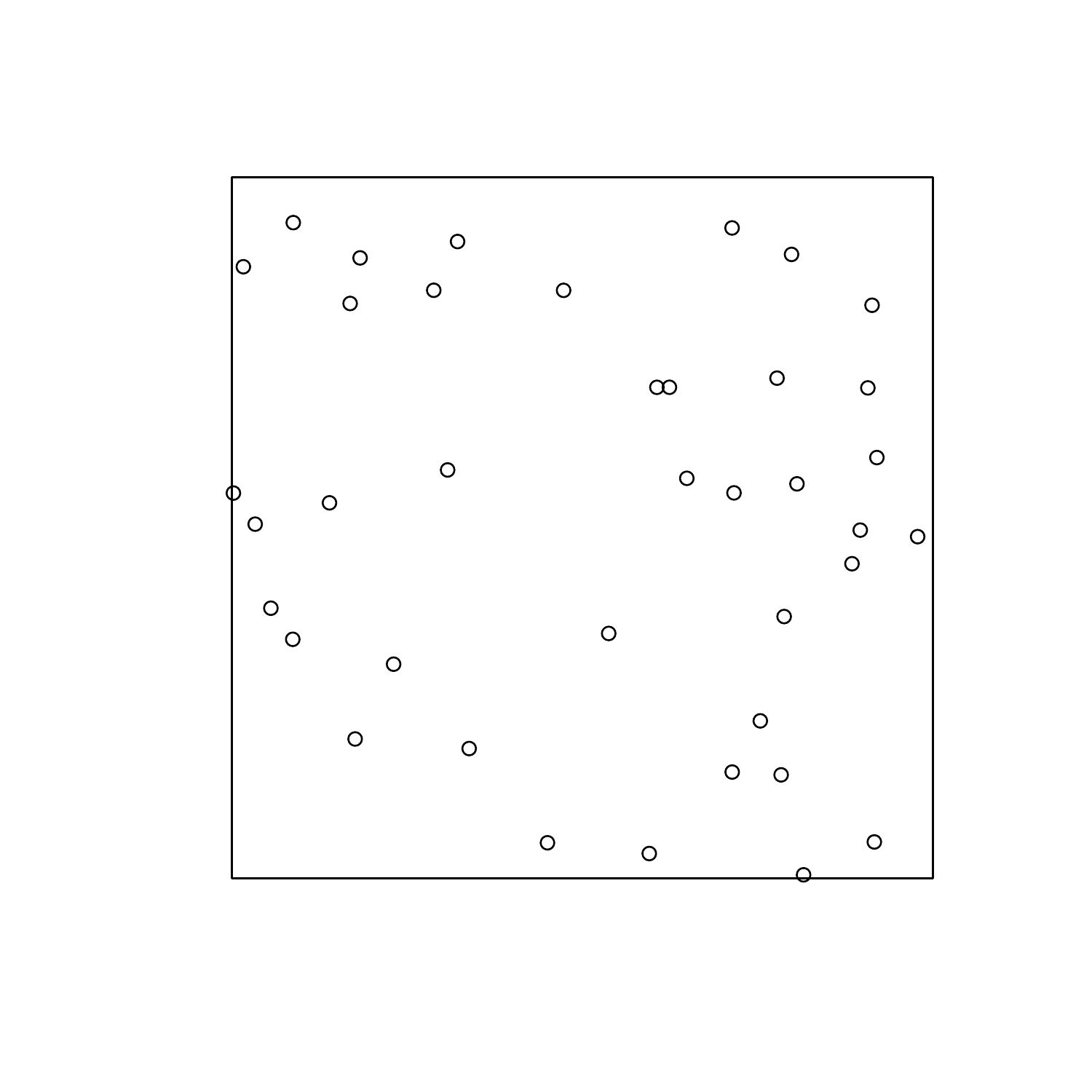}\hspace*{-5.5em}
    \includegraphics[width=0.42\textwidth]{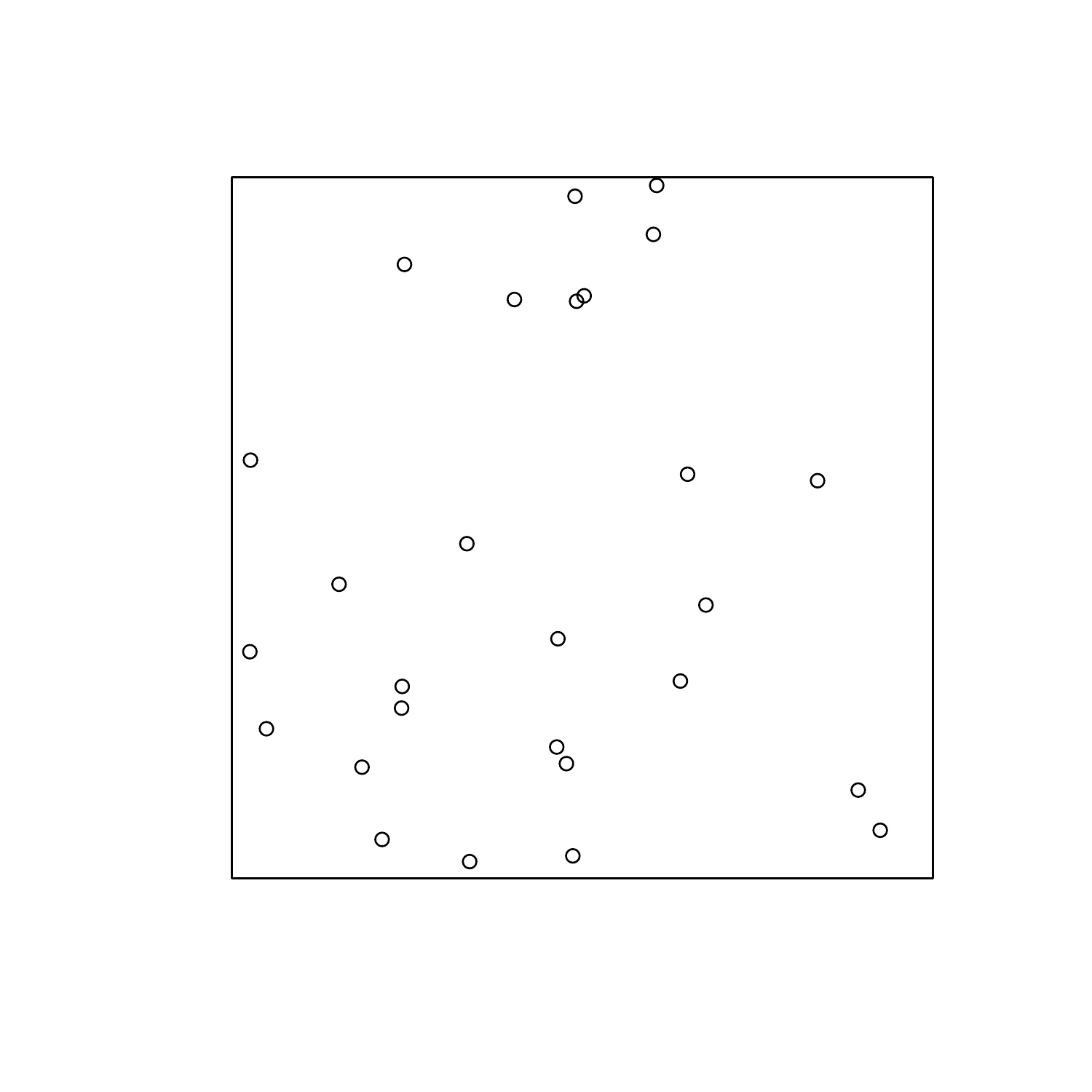}\hspace*{-5.5em}
    \includegraphics[width=0.42\textwidth]{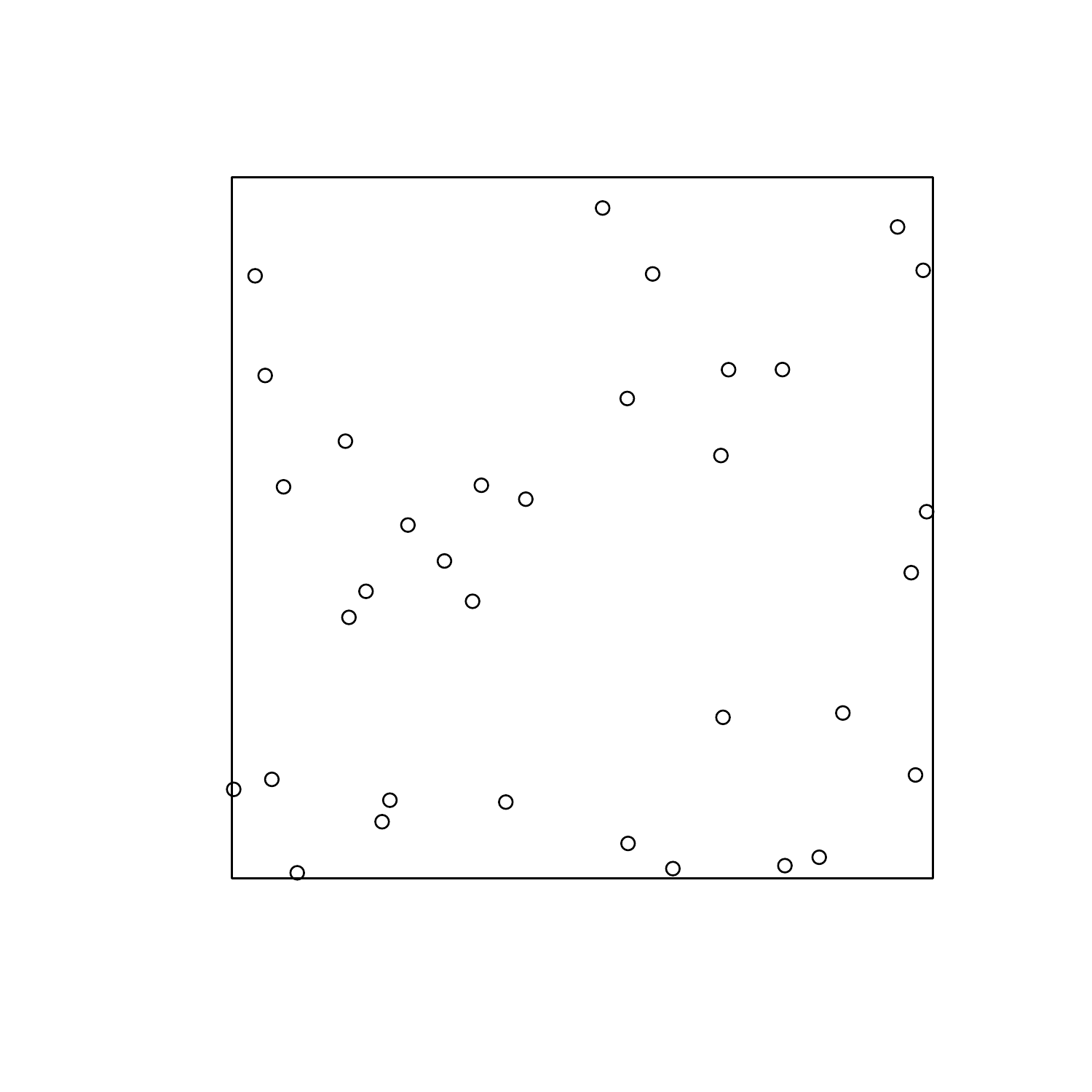}\hspace*{-5em}
    \vspace*{-14mm}

    \caption{Simulations from $\strauss(\beta, \gamma; 0.1)$-distributions, where rowwise from left to right $\gamma = 0,0.2,0.4,0.6,0.8,1$ and $\beta$ is adjusted such that $\lambda=35$. For $\gamma=0$ we have a realization of a hard core process, for $\gamma=1$ a realization from CSR.}
    \label{fig:StraussExamples}
\end{figure}

We perform two different experiments here. In the first one the patterns in Group 1 are sampled from $\csr(35)$ corresponding to $\gamma=1$, in the second one they are sampled from the mentioned hard core process with $\lambda=35$ corresponding to $\gamma=0$. The patterns in Group~2 are sampled in both experiment from each of the six $\gamma$-values in turn. The results are listed in Table~\ref{tab:InteractionResults}.
\begin{table}[!htb]
    \begin{center}
        \begin{tabular}{|c||C{2.5em}|C{2.5em}|C{2.5em}|C{2.5em}|C{2.5em}|C{2.5em}|}
            \hline
            \hspace*{4mm} \scalebox{0.9}[0.95]{$\gamma=1$\hspace*{5mm}vs.\hspace*{-4mm}} &\scalebox{0.9}[0.95]{$\gamma=0$} &\!\scalebox{0.9}[0.95]{$\gamma=0.2$} &\!\scalebox{0.9}[0.95]{$\gamma=0.4$} &\!\scalebox{0.9}[0.95]{$\gamma=0.6$} &\!\scalebox{0.9}[0.95]{$\gamma=0.8$} &\scalebox{0.9}[0.95]{$\gamma=1$} \\
            \hline
            \hline
            Anderson $F_A$ & $98$ & $41$ & $13$ & $8$ & $4$ & $5$  \\
            \hline
            new $L$ & $100$ & $100$ & $95$ & $67$ & $20$ & $3$  \\
            \hline
            Fr\'echet $T_F$ & $100$ & $98$ & $78$ & $45$ & $20$ & $4$  \\
            \hline
            Fr\'echet $T_L$ & $100$ & $99$ & $88$ & $45$ & $20$ & $4$  \\
            \hline
        \end{tabular}
    \end{center}
     \begin{center}
        \begin{tabular}{|c||C{2.5em}|C{2.5em}|C{2.5em}|C{2.5em}|C{2.5em}|C{2.5em}|}
            \hline
            \hspace*{4mm} \scalebox{0.9}[0.95]{$\gamma=0$\hspace*{5mm}vs.\hspace*{-4mm}} &\scalebox{0.9}[0.95]{$\gamma=0$} &\!\scalebox{0.9}[0.95]{$\gamma=0.2$} &\!\scalebox{0.9}[0.95]{$\gamma=0.4$} &\!\scalebox{0.9}[0.95]{$\gamma=0.6$} &\!\scalebox{0.9}[0.95]{$\gamma=0.8$} &\scalebox{0.9}[0.95]{$\gamma=1$} \\
            \hline
            Anderson $F_A$ & $3$ & $55$ & $90$ & $98$ & $97$ & $99$  \\
            \hline
            new $L$ & $11$ & $60$ & $96$ & $100$ & $100$ & $100$  \\
            \hline
            Fr\'echet $T_F$ & $6$ & $57$ & $91$ & $97$ & $100$ & $100$  \\
            \hline
            Fr\'echet $T_L$ & $9$ & $33$ & $82$ & $95$ & $99$ & $100$  \\
            \hline
        \end{tabular}
    \end{center}
    \vspace*{-4mm}
    \caption{Numbers of rejections of the null hypothesis ``equal distribution in both groups'' based on 100 data sets per column. In each data the point patterns in both groups are sampled from a Strauss distribution with $\lambda=35$ and $R=0.1$. The first group is sampled using $\gamma=1$ or $\gamma=0$ as indicated on the top left of the table and the second group uses $\gamma$ as indicated in the column.}
    \label{tab:InteractionResults}
\end{table}

In contrast to the situation in the previous subsection (different inhomogeneity), we now observe that the \emph{indirect} ANOVA procedures, i.e.\ the Levene-type tests perform considerably better than the direct ANOVA procedures. Again this is intuitively understandable because a small $\gamma$ in the Strauss process leads to less dispersion, both in terms of a smaller variance for the total number of points and also with respect to typical distances of points from one another: for small $\gamma$ the points are quite regularly placed, whereas for larger $\gamma$ there are erratic patches that are free of points leading typically to some points that have to be matched over longer distances, which in the squared Euclidean metric has quite some influence.
A small $\gamma$ will also lead to smaller average distances than a larger $\gamma$ (either between point patterns or relative to a barycenter), which may explain why the difference in the performance of the indirect and direct ANOVA tests is somewhat less pronounced than in the inhomogeneity experiment.
\\
Note again that the powers of the tests based on pairwise distances are slightly better than those of the tests based on barycenters.

\section{Applications} \label{sec:realdata}

In this section we apply our Levene's test to a real data example. We investigate the location of bubbles in a mineral flotation experiment.
The structure of the data calls for a two factor design. 
We establish a distance based two-way Levene's test and compare its performance to existing methods. 
The classical two-way ANOVA design can be found in Section~\ref{sec:classic}.

\subsection{Balanced Two-Way Levene's Test}\label{subsec:balancedTwoWayLevene}

As mentioned in Subsection~\ref{subsec:distanceLeveneBasics} it is easy to generalize statistic (\ref{stat:distleveneBalanced}) to a two-way design, that will further be useful for the bubble data analyzed in the next Subsection. 

Suppose we have independent observations $x_{i_1 i_2 j} \in \mcx$, $1\leq j \leq \tn$, $1 \leq i_1 \leq k_1$, $1 \leq i_2 \leq k_2$ from groups obtained by crossing a Factor $a$ with $k_1$ levels and a Factor $b$ with $k_2$ levels with $\tn$ observations for each combination. In a similar way as above we denote by $d_{i_1 i_2 j}$ the $j$-th half-distance in the group $(i_1,i_2)$, where $j = 1, \ldots, \tN:= \binom{\tn}{2}$. Set then
\begin{align*}
    \rss &= \sum_{i_1=1}^{k_1}\sum_{i_2=1}^{k_2}\sum_{j=1}^{\tN} (d_{i_1i_2j}-\bar{d}_{i_1i_2\cdot})^2
    \\
    \mss &= \sum_{i_1=1}^{k_1} \sum_{i_2=1}^{k_2} \tn (\bar{d}_{i_1i_2\cdot}-\bar{d}_{\cdot\cdot\cdot})^2
    \\
    \ssa &= \sum_{i_1=1}^{k_1} k_2 \tn (\bar{d}_{i_1\cdot\cdot}-\bar{d}_{\cdot\cdot\cdot})^2
    \\
    \ssb &= \sum_{i_2=1}^{k_2} k_1 \tn (\bar{d}_{\cdot i_2\cdot}-\bar{d}_{\cdot\cdot\cdot})^2
    \\
    \ssi &= \sum_{i_1=1}^{k_1}\sum_{i_2=1}^{k_2} \tn (\bar{d}_{i_1i_2\cdot} - \bar{d}_{i_1\cdot\cdot} - \bar{d}_{\cdot i_2\cdot} + \bar{d}_{\cdot\cdot\cdot})^2,
\end{align*}
where the various means are taken over the dot components in the usual way. Note that we never use any distances between observations of different factor combinations.

In addition to the omnibus test for group differences as in one-way ANOVA, we may then perform Levene-type tests for effects of Factor a and b separately, as well as for an interaction effect. 
The corresponding statistics are
\begin{align*}
    L = \frac{N - k_1 k_2}{(k_1 k_2 - 1)} \frac{\mss}{\rss}, \hspace{2mm}
    La = \frac{N - k_1 k_2}{k_1 - 1} \frac{\ssa}{\rss}, \hspace{2mm}
    Lb = \frac{N - k_1 k_2}{k_2 - 1} \frac{\ssb}{\rss}, \hspace{2mm}
    Li = \frac{N - k_1 k_2}{(k_1-1)(k_2-1)} \frac{\ssi}{\rss}.
\end{align*}

\subsection{Bubble Data}

\begin{figure}[!htb]
    \includegraphics[width=1\textwidth]{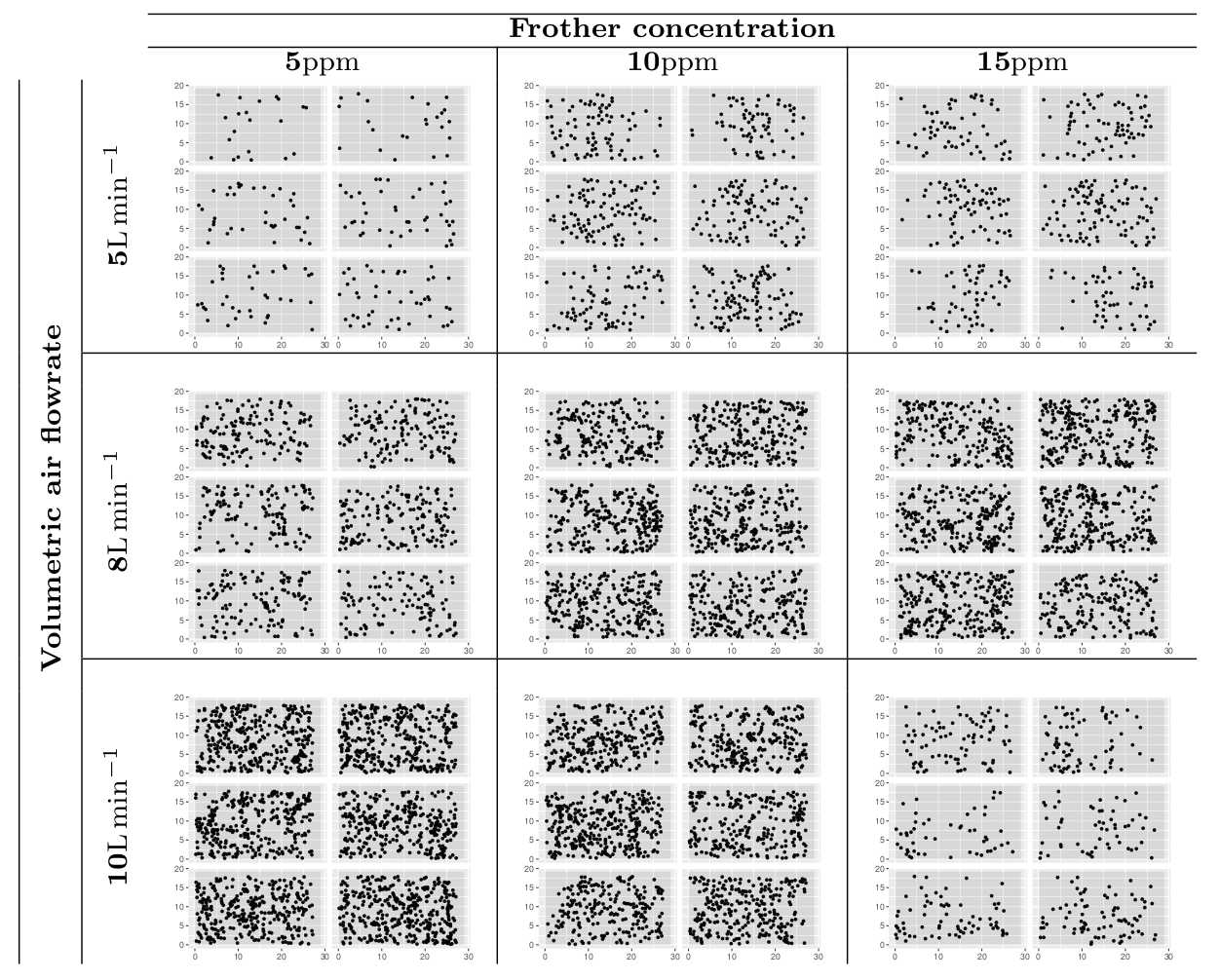}
    \caption{Arrangement of floating bubbles data. Rows represent the three frother concentration levels and columns the three volumetric air flowrate levels (treatments). Each cell contains six spatial point patterns (responses).}
    \label{fig:BubbleData}
\end{figure}

We consider the data from \cite{gonzalez2021two} which provides locations of bubbles in a mineral flotation experiment, where the interest
is analysing if the spatial distribution might be affected by frother concentrations and volumetric airflow rates. Indeed, the data set consists of 54 images containing a total of 8385 floating bubbles. The images of bubbles can be regarded as spatial point patterns where the
centroids of the bubbles correspond to the points. In addition, we have three frother concentration levels (5 ppm, 10 ppm, 15 ppm) 
as well as three volumetric airflow rate levels (5 l/min, 8 l/min, 10 l/min), and we have six replicates of point patterns at each combination of levels of such factors. The treatment combinations of the experiment, as well as the observed bubble point patterns, are represented in Figure~\ref{fig:BubbleData}.

We used the two-way design of Levene's statistic from Section \ref{subsec:balancedTwoWayLevene} to test for influence of the different factors, interaction and differences between the groups. For comparison we also used the two factor statistics from \cite{anderson2001new}, we performed a two factor ANOVA on the number of points per pattern, and finally complemented our analysis with a two factor ANOVA with $K$-functions, so as to link our analysis with that of \cite{gonzalez2021two}. We did a permutation test with $999$ permutations.

In Section \ref{sec:simulstud}, the cutoff was always fixed to $C=0.25$. This was a reasonable value for point patterns with expected $35$ points in the unit square. 
In the bubble data, the number of points per observed pattern ranges from $21$ to $353$. With such a great variability in the number of points we suggest adjusting the cutoff to prevent that distances between two patterns are dominated by their different numbers of points. For the results presented in this section we computed the mean number of points of the tested patterns $\bar{n}$ and used the cutoff $\bar{C} = 0.25 \cdot 35/\bar{n}$ for the computations of $d_{TT}$. For more details to the cutoff see \eqref{eq:ttdef}.

\begin{table}[!htb]
  \begin{center}
      \begin{tabular}{|l||c|c|c|c|}
          \hline
          $p$-values & FC & VA & Interaction & Overall \\
          \hline
          Anderson $F_A$ & 0.003 & 0.001 & 0.001 & 0.001 \\
          \hline
          new $L$ & 0.001 & 0.001 & 0.001 & 0.001 \\
          \hline
          Number of points & 0.001 & 0.001 & 0.001 & 0.001 \\
          \hline
          $K$-functions & 0.005 & 0.001 & 0.001 & 0.001$^{**}$ \\
          \hline
      \end{tabular}
  \end{center}
  $^{**}$ this is the p-value for the sum of both factors, not the overall ANOVA statistic.
  \caption{Results for the different tests for the bubble data. Quantiles are obtained by a per\-mu\-ta\-tion test with $999$ permutations. The cutoff is $C=0.0564$, the maximal radius for the $K$-functions is $r=0.15$.}
  \label{tab:BubbleData}
\end{table}

\begin{table}[!htb]
  \begin{center}
    \begin{tabular}{|l||c|c|c|c|}
        \hline
        $p$-values & FC & VA & Interaction & Overall \\
        \hline
        Anderson $F_A$ & 0.043 & 0.001 & 0.019 & 0.001 \\
        \hline
        new $L$ & 0.001 & 0.001 & 0.001 & 0.001 \\
        \hline
        Number of points & 0.001 & 0.002 & 0.001 & 0.001 \\
        \hline
        $K$-functions & 0.002 & 0.022 & 0.002 & 0.006$^{**}$ \\
        \hline
    \end{tabular}
  \end{center}
  $^{**}$ this is the p-value for the sum of both factors, not the overall ANOVA statistic.
  \caption{Results for the different tests for the bubble data, leaving out the frother concentration of $15$ppm. Quantiles are obtained by a permutation test with $999$ permutations. The cutoff is $C=0.0636$, the maximal radius for the $K$-functions is $r=0.15$.}
  \label{tab:BubbleData_part}
\end{table}

The p-values of the permutation tests are shown in Tables~\ref{tab:BubbleData} and \ref{tab:BubbleData_part}. In particular, Table~\ref{tab:BubbleData} shows results for the whole data set, while Table~\ref{tab:BubbleData_part} depicts results for only part of the data, leaving out the third column, i.e. any patterns from frother concentration of $15$ ppm. In both cases, our new Levene, Anderson $F_A$, the ANOVA on the number of points per pattern, and the ANOVA for $K$-functions detect significant influence of each of the two factors and the interaction.
We already recommended to always perform both, the tests for differences in variability and the test for differences of means. 
In the second test scenario, both Levene's test and Anderson $F_A$ detect significance for the frother concentration and the interaction of both parameters for our usual significance level of $5\%$. 
But the relative difference between the $p$-values of the two tests is very large. 
For the smaller significance level of $1\%$, our Levene's test still detects significance where Anderson $F_A$ does not.
So the test for differences of means might not be enough in a practical application. This is particularly important in cases where, as it is the case for the bubble data, the number of points plays a crucial role in the behavior and structure of the point patterns.

We see that for this data apparently the numbers of points per pattern contain enough information to detect significant influence of the factors. 
This is not very surprising since the number of points per pattern is similar in the $6$ patterns of a single cell, but the differences between cells are large.
\\
This observation is reinforced by a classical multidimensional scaling (mds). 
Based on the TT-distances between the point patterns, we translated every point pattern into a single point in $\mathbb{R}^2$. 
The mds was applied first for the whole bubble data set, see Figure~\ref{fig:cmdwhole}, and then for a subset of the data consisting of the first and second columns, leaving out the data with a frother concentration of $15$ ppm, see Figure~\ref{fig:cmdnofc15}. 
This is the same data that we used for our analyses in Tables~\ref{tab:BubbleData} and~\ref{tab:BubbleData_part}.
The three levels of the air flow are encoded by the colors `red', `green' and `blue', same color means same air flow rate, and the three levels of the frother concentration are encoded by the symbols `circle', `triangle' and `cross'. 
When we compare these plots to the images of the point patterns in Figure~\ref{fig:BubbleData} we can see that the multidimensional scaling sorts the point patterns from left to right in ascending order by their number of points per pattern. 
In Figure~\ref{fig:cmdnofc15} we can see that the points that correspond to the data with a frother concentration of $5$ ppm, i.e. the circles, and the data with a frother concentration of $10$ ppm, i.e. the triangles, are scattered differently.
The (coordinate-wise) means of the triangles and circles are similar, but we can see that the circles are more scattered along both axes.
We conjecture that it is this difference in scatter that our Levene's test is able to detect in Table~\ref{tab:BubbleData_part}, whereas the Anderson $F_A$ only barely detects a slight difference in means.

\begin{figure}[!p]
  \includegraphics[width=0.95\textwidth]{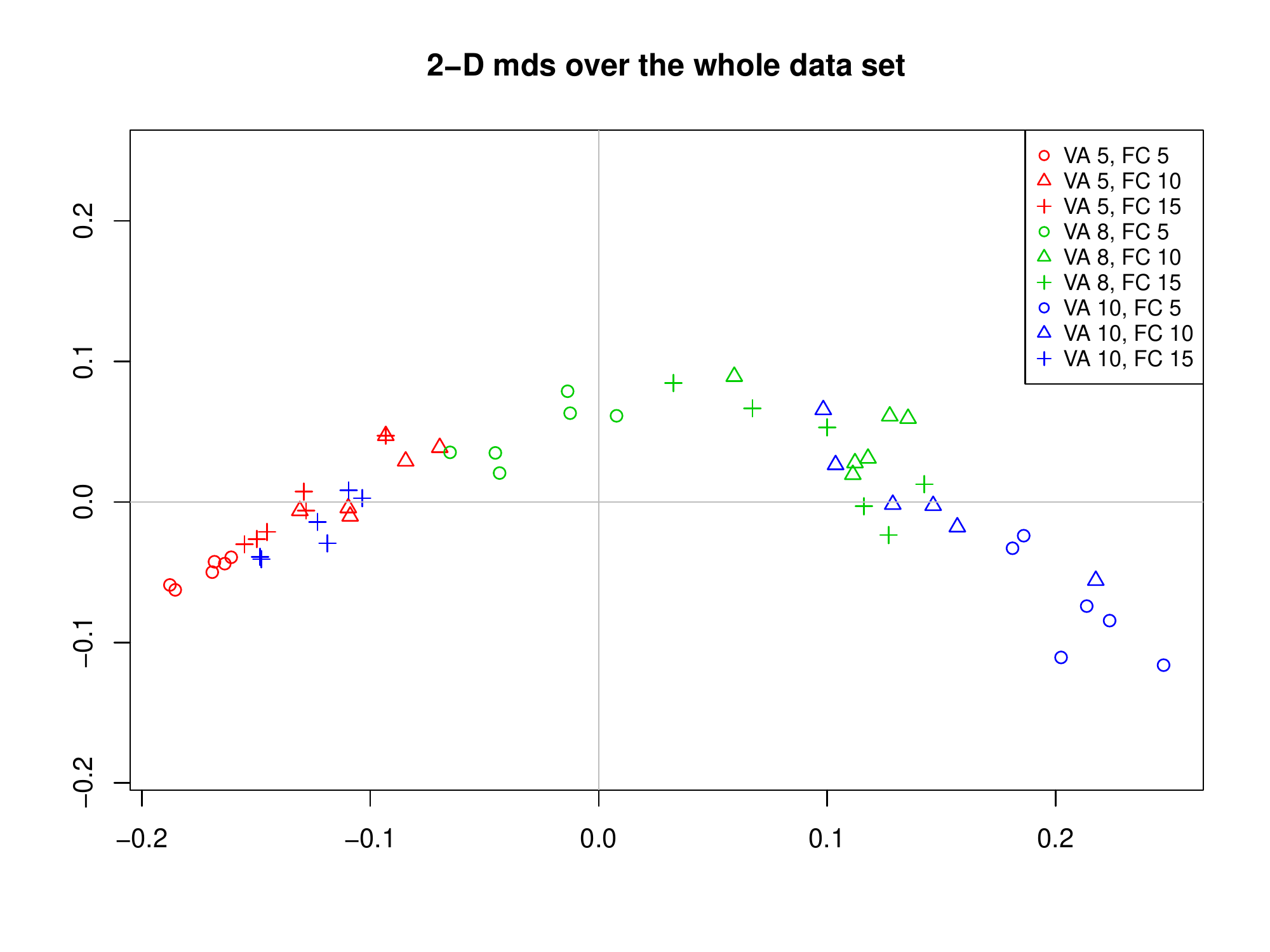}\vspace*{-3em}
  \caption{The bubble data after a multidimensional scaling into two dimensions based on the distance matrix w.r.t the TT-metric.}
  \label{fig:cmdwhole}
\end{figure}
\begin{figure}[!p]
  \includegraphics[width=0.95\textwidth]{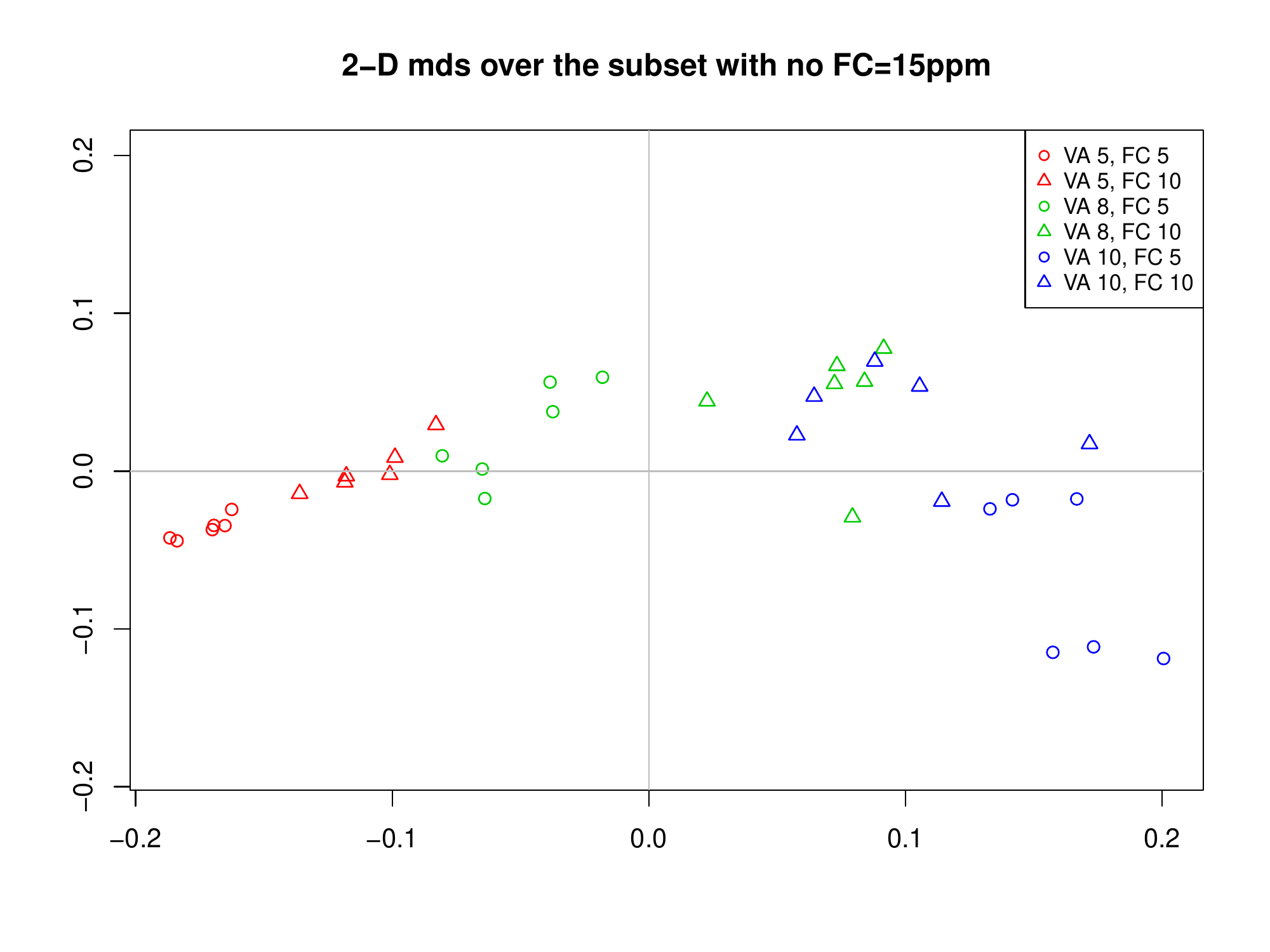}\vspace*{-3em}
  \caption{The bubble data without the third column, i.e. without data where FC=$15$ppm, after a multidimensional scaling into two dimensions based on the distance matrix w.r.t the TT-metric.}
  \label{fig:cmdnofc15}
\end{figure}

\section{Conclusions and Discussion} \label{sec:conclusion}

In this paper we gave an overview of some ANOVA procedures that can be used for data in general metric spaces. 
We introduced a new method that is similar to Levene's test and compared it to existing methods with regard to point pattern data in Section~\ref{sec:simulstud}.
In the studies, see Tables \ref{tab:InhomoResults} and \ref{tab:InteractionResults} for the results, we compared the distance-based ANOVA from \cite{anderson2001new}, the distance-based Levene's test, \eqref{stat:distleveneBalanced}, introduced in this paper and the tests based on the ANOVA statistic $T_F$ and the Levene statistic $T_L$ from \cite{dubey2019frechet}.
\\ 
The latter proposed in their paper the combined statistic $T = T_L + T_F$. 
In our simulations we wanted to put a focus on the two fundamentally different ways of ``location'' and ``dispersion'' in which group distributions can differ, even in an abstract metric space.
We therefore did our tests with the statistics $T_L$ and $T_F$ seperately, which allows us to have a direct comparison between the distance-based statistics and the statistics of \cite{dubey2019frechet}.
We also did the tests with the proposed combined statistic $T$, and for completeness we give the results in Tables~\ref{tab:SumstatResults_Inhomogeneity} and \ref{tab:SumstatResults_Interaction}.
In the tests for differences in interaction, comparing Tables~\ref{tab:SumstatResults_Interaction} and \ref{tab:InteractionResults}, we can see that the performance of the statistic $T$ is ``between'' the performance of $T_L$ and $T_F$.
For the tests of inhomogeneity, see Tables~\ref{tab:SumstatResults_Inhomogeneity} and \ref{tab:InhomoResults}, the combined statistic $T$ performs almost as good as the better statistic of $T_L$ and $T_F$, which is in this case $T_F$.
\\
For the presented scenarios and the chosen parameters all the statistics worked well for their designated purpose, in particular the ANOVA statistics for detecting inhomogeneity and the Levene's statistics for detecting differences in point interaction. 
But for different scenarios this might not be the case. For a cutoff of $C=0.1$ instead of the proposed $C=0.25$, we observed that the statistic $T$ and the Anderson $F_A$ do not work as well anymore in detecting differences in interaction, see Tables~\ref{tab:PowerPoissons} and \ref{tab:PowerStrauss}. 
The performance of $T$ is considerably worse than the performance of $T_L$ and $T_F$ is working even more poorly as well. 
The distance based ANOVA statistic of Anderson also performs very poorly compared to the cutoff of $C=0.25$, while our statistic $L$ is working even better with the smaller cutoff.
\medskip

For future research it would be interesting to take a closer look at the (co)variance estimator $\gamma$ of our $\widetilde{L}$ statistic. This estimator is not unbiased and it remains open if a statistic with an unbiased estimator works even better, in particular for the asymptotics.
\\
There are also more complex designs for the $2$-factor ANOVA, with different sums of squares or designs that allow for different group sizes. Our $L$ statistic could be generalized to more complex $2$-factor designs, or even $k$-factor designs.
\\
Additionally it would be interesting to test our statistics with more kinds of data. On the one hand there are different kinds of point pattern data, e.g. marked point patterns, but also data from other metric spaces, e.g. image data or graph data.
\medskip

In Section~\ref{sec:simulstud} we already mentioned that the computation of the Fr\'echet $T$ statistic is more time consuming than the distance based tests, because of the barycenter calculation. 
If we take for example the data from Section~\ref{subsec:simulstud_interact}, the distance based Anderson $F_A$ and our new $L$ take about $8$ seconds and $2$ seconds, respectively, for $100$ permutation tests with $999$ permutations each. 
The calculation of $100$ permutation tests of Fr\'echet $T$ with $99$ permutations each takes about $45$ minutes. 
The workhorse computation of all tests is done in C++. However, parts of the overhead for Anderson $F_A$ and Fr\'echet $T$ are programmed in R and a complete implementation in C++ might improve the runtime a little bit. 
These numbers are merely meant to give a general impression of the order of magnitude of the runtimes. 
The distance based tests take only a few seconds, while for the Fr\'echet tests the computation of the barycenters takes many minutes, even with the fast heuristics instead of the exact solution and with fewer permutations. 

Overall we find that the new $L$ in combination with the Anderson $F_A$ has a similar performance and allows for considerably faster computation than the other methods in settings where the computation of barycenters is costly.

\begin{table}[htb]
    \begin{center}
        \begin{tabular}{|c||C{1.5em}|C{1.5em}|C{1.5em}|C{1.5em}|C{1.5em}|C{1.5em}|C{1.5em}|}
            \hline
            Scenario & $1$ & $2$ & $3$ & $4$ & $5$ & $6$ & $0$ \\
            \hline
            \hline
            Fr\'echet $T$ & 100 & 100 & 100 & 98 & 11 & 2 & 6 \\
             \hline
        \end{tabular}
    \end{center}
    \vspace*{-4mm}
    \caption{Performance of the sum statistic Fr\'echet $T$. Numbers of rejections of the null hypothesis ``equal distribution in both groups'' based on 100 data sets per column for the $7$ scenarios of inhomogeneity.}
   \label{tab:SumstatResults_Inhomogeneity}
\end{table}
\begin{table}[htb]
  \begin{center}
      \begin{tabular}{|c||C{2.5em}|C{2.5em}|C{2.5em}|C{2.5em}|C{2.5em}|C{2.5em}|}
          \hline
          \hspace*{4mm} \scalebox{0.9}[0.95]{$\gamma=1$\hspace*{2.5mm}vs.\hspace*{-1.5mm}} &\scalebox{0.9}[0.95]{$\gamma=0$} &\!\scalebox{0.9}[0.95]{$\gamma=0.2$} &\!\scalebox{0.9}[0.95]{$\gamma=0.4$} &\!\scalebox{0.9}[0.95]{$\gamma=0.6$} &\!\scalebox{0.9}[0.95]{$\gamma=0.8$} &\scalebox{0.9}[0.95]{$\gamma=1$} \\
          \hline
          \hline
          Fr\'echet $T$ & $100$ & $98$ & $77$ & $48$ & $16$ & $2$  \\
          \hline
      \end{tabular}

      \begin{tabular}{|c||C{2.5em}|C{2.5em}|C{2.5em}|C{2.5em}|C{2.5em}|C{2.5em}|}
          \hline
          \hspace*{4mm} \scalebox{0.9}[0.95]{$\gamma=0$\hspace*{2.5mm}vs.\hspace*{-1.5mm}} &\scalebox{0.9}[0.95]{$\gamma=0$} &\!\scalebox{0.9}[0.95]{$\gamma=0.2$} &\!\scalebox{0.9}[0.95]{$\gamma=0.4$} &\!\scalebox{0.9}[0.95]{$\gamma=0.6$} &\!\scalebox{0.9}[0.95]{$\gamma=0.8$} &\scalebox{0.9}[0.95]{$\gamma=1$} \\
          \hline
          \hline
          Fr\'echet $T$ & $5$ & $48$ & $91$ & $98$ & $100$ & $100$  \\
          \hline
      \end{tabular}
  \end{center}
  \vspace*{-4mm}
  \caption{Performance of the sum statistic Fr\'echet $T$. Numbers of rejections of the null hypothesis ``equal distribution in both groups'' based on 100 data sets per column for the different scenarios of interaction between points.}
 \label{tab:SumstatResults_Interaction}
\end{table}

\begin{table}[!htb]
    \begin{center}
      \begin{tabular}{|c||C{2.5em}|C{2.5em}|C{2.5em}|C{2.5em}|C{2.5em}|C{2.5em}|}
        \hline
        \hspace*{4mm} \scalebox{0.9}[0.95]{$\gamma=1$\hspace*{5mm}vs.\hspace*{-4mm}} &\scalebox{0.9}[0.95]{$\gamma=0$} &\!\scalebox{0.9}[0.95]{$\gamma=0.2$} &\!\scalebox{0.9}[0.95]{$\gamma=0.4$} &\!\scalebox{0.9}[0.95]{$\gamma=0.6$} &\!\scalebox{0.9}[0.95]{$\gamma=0.8$} &\scalebox{0.9}[0.95]{$\gamma=1$} \\
            \hline
            Anderson $F_A$ & 37 & 16 & 8 & 5 & 12 & 3 \\
            \hline
            new $L$ & 100 & 100 & 96 & 69 & 15 & 6 \\
            \hline
             Fr\'echet $T_F$ & 49 & 31 & 10 & 4 & 5 & 9 \\
            \hline           
            Fr\'echet $T_L$ & 100 & 99 & 87 & 47 & 13 & 7 \\
            \hline
            Fr\'echet $T$ & 99 & 84 & 38 & 12 & 7 & 9 \\
            \hline
        \end{tabular}
    \end{center}
    \caption{C=0.1, The first scenario: Groups of Poisson-distributed point patterns vs groups of Strauss-distributed point patterns with $6$ different gammas: $\gamma = 0,0.2,0.4,0.6,0.8,1$, $R=0.1$, $\lambda=35$, $\alpha = 0.05$, $20$ patterns per group. Numbers indicate how many times the hypothesis "equal distributions in both groups" is rejected out of $100$ times. The tests should see no difference between groups of Poisson-patterns and Strauss-patterns with $\gamma=1$.}
    \label{tab:PowerPoissons}
\end{table}
\begin{table}[!htb]
    \begin{center}
      \begin{tabular}{|c||C{2.5em}|C{2.5em}|C{2.5em}|C{2.5em}|C{2.5em}|C{2.5em}|}
        \hline
        \hspace*{4mm} \scalebox{0.9}[0.95]{$\gamma=0$\hspace*{5mm}vs.\hspace*{-4mm}} &\scalebox{0.9}[0.95]{$\gamma=0$} &\!\scalebox{0.9}[0.95]{$\gamma=0.2$} &\!\scalebox{0.9}[0.95]{$\gamma=0.4$} &\!\scalebox{0.9}[0.95]{$\gamma=0.6$} &\!\scalebox{0.9}[0.95]{$\gamma=0.8$} &\scalebox{0.9}[0.95]{$\gamma=1$} \\
            \hline
            Anderson $F_A$ & 5 & 11 & 35 & 41 & 40 & 50 \\
            \hline
            new $L$ & 5 & 99 & 100 & 100 & 100 & 100 \\
            \hline
            Fr\'echet $T_F$ & 8 & 14 & 26 & 41 & 40 & 53 \\
            \hline
            Fr\'echet $T_L$ & 7 & 87 & 100 & 100 & 100 & 100 \\
            \hline
            Fr\'echet $T$ & 7 & 31 & 73 & 95 & 100 & 100 \\
            \hline
        \end{tabular}
    \end{center}
    \caption{C=0.1, The second scenario: Groups of Strauss-distributed point patterns with a fixed $\gamma=0$ vs groups of Strauss-distributed point patterns with $6$ different gammas: $\gamma = 0,0.2,0.4,0.6,0.8,1$, $R=0.1$, $\lambda=35$, $\alpha = 0.05$, $20$ patterns per group. Numbers indicate how many times the hypothesis "equal distributions in both groups" is rejected out of $100$ times. The tests should see no difference for $\gamma=0$.}
    \label{tab:PowerStrauss}
\end{table}

\bibliographystyle{apalike}
\bibliography{Anova}

\newpage

\appendix

\section{Auxiliary Results Used for the Proof of Theorem~\ref{thm:mainasymp}}

For completeness and self-containedness we state here (consequences of) results from the literature as well as some additional calculations needed for the proof of Theorem~\ref{thm:mainasymp}.

Firstly we formulate a straightforward generalization of Hoeffding's theorem for the asymptotic normality of $U$-statistics (univariate version of Theorem~7.1 in \citealp{hoeffding1948class}) for random elements in the general metric space $\mcx$ with countably generated Borel $\sigma$-algebra. See also Theorem 1(b) of~\cite{denker1983u}, where this result is further generalized to (weakly) dependent sequences of random elements.

\begin{theo} \label{theo:hoeffding}
    Let $(X_n)_{n \in \mathbb{N}}$ be an i.i.d.\ sequence of $\mathcal{X}$-valued random elements. Let $h \colon \mathcal{X}^m \to \RR$ be symmetric and non-degenerate in the sense that there are $x_2, \ldots, x_m \in \mathcal{X}$ such that
    \begin{equation*}
        \mathbb{E} h(X_1, x_2, \ldots, x_m) \neq 0.
    \end{equation*}
    Suppose further that $\mathbb{E} \bigl(h(X_1,\ldots,X_m)^2\bigr) < \infty$.
    We write
    \begin{equation*}
        U_n = {\binom{n}{m}}^{-1} \sum_{\substack{i_1,\ldots,i_m=1 \\ i_1 < \ldots < i_m}}^{n} h(X_{i_1},\ldots,X_{i_m}).
    \end{equation*}
    for the $U$-statistic with kernel $h$.
    Then
    \begin{equation*}
        \sqrt{n} (U_n - \mathbb{E}(U_n)) \stackrel{\mathcal{D}}{\longrightarrow} \mathcal{N}(0,m^2 \gamma_h^2),
    \end{equation*}
    where for an independent copy $(\tX_2,\ldots,\tX_m)$ of $(X_2,\ldots,X_m)$
    \begin{equation*}
        \gamma_h^2 = \mathrm{Cov} \bigl(h(X_1,X_2,\ldots,X_m), h(X_1,\tX_2,\ldots,\tX_m) \bigr) = \mathrm{Var} \bigl(\mathbb{E} (h(X_1,\ldots,X_m) \, \vert \, X_1) \bigr).
    \end{equation*}
\end{theo}
\begin{remark}
In the setting of Theorem~\ref{theo:hoeffding} above, Theorem~5.2 of~\cite{hoeffding1948class} yields
\begin{equation*}
   m^2 \gamma_h^2 \leq n \mathrm{Var}(U_n) \leq m \mathrm{Var}(h(X_1,\ldots,X_m))
\end{equation*}
for all $n \geq m$. The right hand bound is sharp for $n=m$ and $n \mathrm{Var}(U_n) \searrow m^2 \gamma_h^2$ as $n \to \infty$.
\end{remark}
The above inequality means in particular that for finite $n$ the expression $\frac{m^2}{n} \gamma_h^2$ can only underestimate $\mathrm{Var}(U_n)$. The exact formula for $m=2$ is
\begin{equation*}
  n \mathrm{Var}(U_n) = \frac{n-2}{n-1} \cdot 4 \gamma_h^2 + \frac{1}{n-1} \cdot 2 \mathrm{Var}(h(X_1,X_2)).
\end{equation*}

The next result is similar to classical ANOVA. For completeness we give its proof.
\begin{lemma} \label{lem:idemcomp}
  Let $C \in \mathbb{R}^{(k-1)\times k}$ as in \eqref{matrixC}, $D\in \mathbb{R}^{n\times k}$ as in \eqref{matrixD}, $U=(u_1,\ldots,u_k)'$ and $\nu = (n_1, \ldots, n_k)'$. We have
  \begin{equation*}
    C'(C(D'D)^{-1}C')^{-1} C = D'D - \frac{1}{n} \nu\nu'
  \end{equation*}
  and
  \begin{equation*}
    U'(D'D - \frac{1}{n} \nu\nu')U = \frac{1}{n}\sum_{i=1}^{k-1}\sum_{j=i+1}^k n_in_j(u_i-u_j)^2
  \end{equation*}
\end{lemma}

\begin{proof} Define
  \begin{align*}
    \nu_{(i)} &:= (n_1,\ldots,n_i)' \text{ , } \; \Lambda_{(i)} := \mathrm{diag}(\nu_{(i)}) \in \RR^{i \times i}
    \; \text{ and } \; \mathds{1}_{(i)} := (1,\ldots,1)' \in \mathbb{R}^i.
  \end{align*}
Then $\mathds{1}_{(i)}\mathds{1}_{(i)}'$ is the $i\times i$ matrix of $1$'s.
  We build up the equality step by step. Since $D'D = \Lambda_{(k)}$ and therefore 
  \begin{align*}
    (D'D)^{-1} = (\Lambda_{(k)})^{-1} = \mathrm{diag}(1/n_1,\ldots,1/n_k),
  \end{align*}
  We obtain 
  \begin{align*}
    C(D^{\prime} D)^{-1}C^{\prime} &= (\Lambda_{(k-1)})^{-1} + \frac{1}{n_k}\cdot \mathds{1}_{(k-1)}\mathds{1}_{(k-1)}'
  \end{align*}
  and 
  \begin{align*}
    (C(D^{\prime} D)^{-1}C^{\prime})^{-1} &= \Lambda_{(k-1)} - \frac{1}{n}\cdot\nu_{(k-1)}\nu_{(k-1)}'
  \end{align*}
  and finally 
  \begin{align*}
    C^\prime(C(D^{\prime} D)^{-1}C^{\prime})^{-1}C = \Lambda_{(k)} - \frac{1}{n}\cdot\nu\nu'
  \end{align*}

  When we multiply the vector $U$ from left and right, the $ij$-th entry in the matrix is the coefficient of $u_iu_j$. This leads to
  \begin{align*}
    U'(D'D &- \frac{1}{n} \nu\nu')U
    \\
    &=\sum_{i=1}^k n_i u_i^2 - \frac{1}{n}\sum_{i=1}^k n_i^2u_i^2 - \frac{1}{n}\sum_{i=1}^{k-1}\sum_{j=i+1}^k 2n_in_j u_iu_j
    \\
    &= \frac{1}{n}\sum_{i=1}^k n_i u_i^2 \sum_{j=1}^k n_j - \frac{1}{2n}\sum_{i=1}^k n_i^2(u_i^2 + u_i^2) - \frac{1}{n}\sum_{i=1}^{k-1}\sum_{j=i+1}^k 2n_in_j u_iu_j
    \\
    &= \frac{1}{2n}\sum_{i=1}^k \sum_{j=1}^k n_in_j (u_i^2 + u_j^2) - \frac{1}{2n}\sum_{i=1}^k n_i^2(u_i^2 + u_i^2) - \frac{1}{n}\sum_{i=1}^{k-1}\sum_{j=i+1}^k 2n_in_j u_iu_j
    \\
    &= \frac{1}{n}\sum_{i=1}^{k-1} \sum_{j=i+1}^k n_in_j (u_i^2 + u_j^2) - \frac{1}{n}\sum_{i=1}^{k-1}\sum_{j=i+1}^k 2n_in_j u_iu_j
    \\
    &= \frac{1}{n}\sum_{i=1}^{k-1} \sum_{j=i+1}^k n_in_j (u_i - u_j)^2.
\end{align*}
\end{proof}

\begin{remark} \label{rem:idemcomp}
Let $\bar{u} = \frac{1}{k}\sum_{i=1}^k u_i$. An equivalent expression for $U'(D'D - \frac{1}{n} \nu\nu')U$ in Lemma~\ref{lem:idemcomp} can be computed as
\begin{align*}
  \frac{1}{n}&\sum_{i=1}^{k-1} \sum_{j=i+1}^k n_in_j (u_i - u_j)^2
  \\
  &= \frac{1}{2n}\sum_{i=1}^{k}\sum_{j=1}^k n_in_j ((u_i - \bar{u}) + (\bar{u} - u_j))^2
  \\
  &= \frac{1}{2n}\sum_{i=1}^{k}\sum_{j=1}^k n_in_j (u_i - \bar{u})^2
  +\frac{1}{2n}\sum_{i=1}^{k}\sum_{j=1}^k n_in_j (\bar{u} - u_j)^2
  +\frac{1}{2n}\sum_{i=1}^{k}\sum_{j=1}^k n_in_j (u_i - \bar{u})(\bar{u} - u_j)
  \\
  &= \sum_{i=1}^{k}n_i(u_i - \bar{u})^2
  +\frac{1}{2n}\sum_{i=1}^{k}n_i(u_i - \bar{u})\sum_{j=1}^k n_j (\bar{u} - u_j)
  \\
  &= \sum_{i=1}^{k}n_i(u_i - \bar{u})^2
  +\frac{1}{2n}\left(\sum_{i=1}^{k}n_i(u_i - \bar{u})\right)^2
\end{align*}
If $n_1=\ldots=n_k=\tn$, we see directly from the right-hand side that
\begin{align}\label{eq:varidentity}
  \frac{1}{n}\sum_{i=1}^{k-1} \sum_{j=i+1}^k &n_in_j (u_i - u_j)^2 = \sum_{i=1}^{k}n_i(u_i - \bar{u})^2
  = \tn \sum_{i=1}^{k} (u_i - \bar{u})^2.
\end{align}
\end{remark}

The following lemma is well known. It follows by spectral decomposition, see e.g.\ Kent, Mardia and Bibby (1979), 
Theorem~3.4.4(b), setting $p=1$ and $\Sigma = I$.
\begin{lemma} \label{lem:cochran}
  Let $Z \sim \mcn_n(0,I)$ and let $C \in \RR^{n\times n}$ be symmetric and idempotent. Then $Z' C Z \sim \chi^2_r$, where $r = \mathrm{trace}(C) = \mathrm{rank}(C)$.
\end{lemma}

\end{document}